\documentclass[letterpaper,onecolumn,twoside]{IEEEtran}
\usepackage{amsmath, amssymb, graphicx, psfrag, cite, dsfont}

\newlength{\widthA}
\newcommand{\beq}{\begin{equation}}
\newcommand{\eeq}{\end{equation}}
\newcommand{\beqa}{\begin{eqnarray}}
\newcommand{\eeqa}{\end{eqnarray}}
\newcommand{\beqan}{\begin{eqnarray*}}
\newcommand{\eeqan}{\end{eqnarray*}}

\def\var{\mathop{\mathrm{var}}}

\def\Z{{\mathbb{Z}}}
\def\kt{{\widetilde{k}}}
\def\Ct{{\widetilde{C}}}
\def\G{{\mathcal{G}}}
\def\F{{\mathcal{F}}}

\renewcommand{\P}[1]{\textbf{P}\left({#1}\right)} 
\newcommand{\E}[1]{\textbf{E}\left[{#1}\right]}   

\newcommand{\eqlabel}[1]{ \stackrel{(#1)}{=} }
\newcommand{\geqlabel}[1]{ \stackrel{(#1)}{\geq} }
\newcommand{\leqlabel}[1]{ \stackrel{(#1)}{\leq} }

\newtheorem{theorem}{Theorem}
\newtheorem{lemma}[theorem]{Lemma}
\newtheorem{proposition}[theorem]{Proposition}
\newtheorem{corollary}[theorem]{Corollary}
\newtheorem{definition}{Definition}

\newcommand{\qed}{\nobreak \ifvmode \relax \else
      \ifdim\lastskip<1.5em \hskip-\lastskip
      \hskip1.5em plus0em minus0.5em \fi \nobreak
      \vrule height0.75em width0.5em depth0.25em\fi}
      
\title{The Porosity of Additive Noise Sequences}
\author{Vinith~Misra and
        Tsachy~Weissman,~\IEEEmembership{Senior~Member,~IEEE}
\thanks{This work was supported in part by the NDSEG fellowship, the Stanford Graduate Fellowship, and the NSF SCOI center.}%
\thanks{ V. Misra (email: vinith@stanford.edu),
        and T. Weissman (email: tsachy@stanford.edu) are with
        the Department of Electrical Engineering,
        Stanford University,
        Stanford, CA 94305 USA\@.}}

\begin{document}
\maketitle

\begin{abstract}
Consider a binary additive noise channel with noiseless feedback.  
When the noise is a stationary and ergodic process $\mathbf{Z}$, the 
capacity is $1-\mathbb{H}(\mathbf{Z})$ ($\mathbb{H}(\cdot)$ denoting
the entropy rate).  
It is shown analogously that
when the noise is a deterministic sequence $z^\infty$,
the capacity under finite-state encoding and decoding is
$1-\overline{\rho}(z^\infty)$, where $\overline{\rho}(\cdot)$ is Lempel and Ziv's
finite-state compressibility.  This quantity is termed the \emph{porosity}
$\underline{\sigma}(\cdot)$ of an
individual noise sequence.  A sequence of schemes
are presented that universally achieve porosity for any noise sequence.
These converse and achievability results may be interpreted both as a channel-coding counterpart
to Ziv and Lempel's work in universal source coding, as well as an extension to the work
by Lomnitz and Feder and Shayevitz and Feder on communication across
modulo-additive channels.
Additionally, a slightly more practical architecture
is suggested that draws a connection with finite-state predictability,
as introduced by Feder, Gutman, and Merhav.

\end{abstract}


\section{Introduction}
\label{sec:intro}

The ``core'' results of information theory, starting with Shannon's source
and channel coding theorems, are concerned with probabilistic systems
of a known model: an iid Bernoulli$(1/4)$ source must be compressed, or
perhaps bits are to be communicated across an AWGN channel of known SNR.
One may seek additional generality by asking that a coding scheme simultaneously function
for an entire class of such probabilistic models.  In the case of source coding,
Ziv and Lempel \cite{ZivL1977, ZivL1978} and Ziv \cite{Ziv1978} take this question to its logical
 extreme and ask that a compressor
not only achieve the optimal rate for any probabilistic source model, but do so
for \emph{any individual source sequence}.  In \cite{ZivL1978}, it is discovered that the traditionally relevant probabilistic measurement 
--- entropy rate --- generalizes into a measure for
an individual sequence --- compressibility.  In this paper, an analogous set of questions
yield an analogous set of answers in the context of noisy channel coding with feedback.

Historically, far more attention has been paid to the issue of universality in source coding
than in channel coding.  The source of this discrepancy is readily apparent from Figs. \ref{fig:source-coding} 
and \ref{fig:channel-coding}.  The encoder of Fig. \ref{fig:channel-coding} never observes the noise sequence in any way,
 and so its codebook cannot be dynamically customized to suit the channel.  The source encoder of Figure \ref{fig:source-coding} on
 the other hand has direct access to the source sequence and can therefore adjust to its statistics.
 As such, the degree of universality that can be requested in the classical channel-coding setup 
 is far more restricted than in source coding.  Certainly, this does not preclude
 discussion of ``universality,'' but the term must take on a considerably looser meaning, as
 is discussed in  Sec. \ref{sec:related-work}.
 
 \begin{figure}
 \centering
  \psfrag{s}[cc][cc]{\small Unknown Source}
  \psfrag{m}[cc][cc]{$X^\infty$}
  \psfrag{e}[cc][cc]{E}
  \psfrag{d}[cc][cc]{D}
  \psfrag{m2}{$\widehat{X}^\infty$}
  \includegraphics[width=3in]{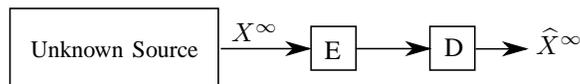}
  \caption{The model for universal source coding.  An unknown source is
  provided to an encoder, which must describe it to the decoder.}
 \label{fig:source-coding}
\end{figure}

\begin{figure}
 \centering
  \psfrag{c}[cc][cc]{\small Unknown Channel}
  \psfrag{m}[cc][cc]{$M^\infty$}
  \psfrag{e}[cc][cc]{E}
  \psfrag{d}[cc][cc]{D}
  \psfrag{m2}[cc][cc]{$\widehat{M}^\infty$}
  \psfrag{x}[cc][cc]{$x_i$}
  \psfrag{y}[cc][cc]{$y_i$}
  \includegraphics[width=3.8in]{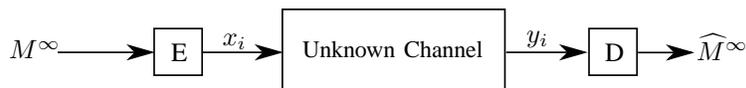}
  \caption{The model for universal channel coding.  A message must
  be communicated over an unknown channel.}
 \label{fig:channel-coding}
\end{figure}
 
The playing field is considerably leveled by introducing a noiseless feedback link, as in Fig. \ref{fig:channel-coding-feedback}.
In particular, the modulo-additive channel of Fig. \ref{fig:channel-coding-additive} allows for a clear and precise analogy to the universal
source coding of Lempel and Ziv.  To highlight some of the parallels:

\begin{figure}
 \centering
  \psfrag{c}[cc][cc]{\small Unknown Channel}
  \psfrag{m}[cc][cc]{$M^\infty$}
  \psfrag{e}[cc][cc]{E}
  \psfrag{d}[cc][cc]{D}
  \psfrag{m2}[cc][cc]{$\widehat{M}^\infty$}
  \psfrag{x}[cc][cc]{$x_i$}
  \psfrag{y}[cc][cc]{$y_i$}
  \psfrag{z}[cc][cc]{$z^{-1}$}
  \psfrag{y-1}[cc][cc]{$y_{i-1}$}
  \includegraphics[width=3.8in]{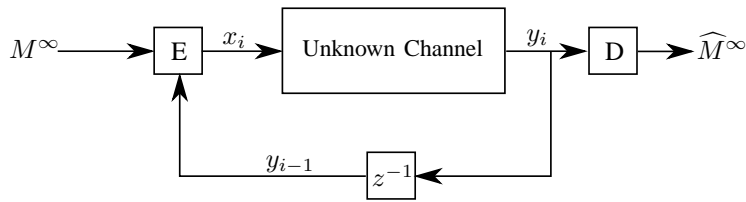}
  \caption{Universal channel coding with noiseless feedback.}
 \label{fig:channel-coding-feedback}
\end{figure}

\begin{figure}
 \centering
  \psfrag{c}[cc][cc]{\small Unknown Channel}
  \psfrag{m}[cc][cc]{$M^\infty$}
  \psfrag{e}[cc][cc]{E}
  \psfrag{d}[cc][cc]{D}
  \psfrag{m2}[cc][cc]{$\widehat{M}^\infty$}
  \psfrag{x}[cc][cc]{$x_i$}
  \psfrag{y}[cc][cc]{$y_i$}
  \psfrag{z}[cc][cc]{$z^{-1}$}
  \psfrag{y-1}[cc][cc]{$y_{i-1}$}
  \psfrag{2}[cc][cc]{$2$}
  \psfrag{noise}[cc][cc]{$z^{\infty}$}
  \psfrag{+}[cc][cc]{$+$}
  \includegraphics[width=3.8in]{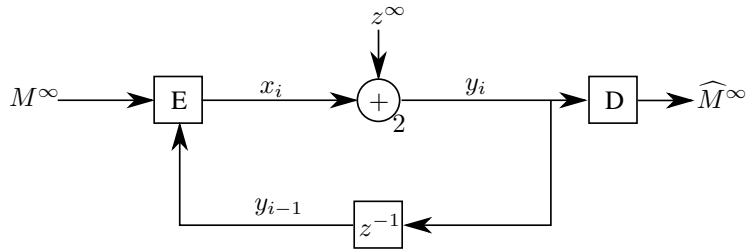}
  \caption{An additive-noise channel with noiseless feedback.}
 \label{fig:channel-coding-additive}
\end{figure}

\begin{itemize}
\item Individual sequences.  In the source coding setting, Lempel and Ziv replace the source random process $X^\infty$ with a 
deterministic sequence $x^\infty$.  Here, the standard stochastic description of channel noise
$Z^\infty$ is supplanted by a specific individual sequence $z^\infty$.
\item Finite state constraint.  Lempel and Ziv ask the question: how well can a source encoder/decoder
perform for a specific individual sequence, if the engineer designing the encoder/decoder \emph{knows} the sequence
ahead of time?
Clearly, if the encoder and decoder are unconstrained, this problem trivializes: 
one may design a decoder that, with absolutely no input from the encoder, produces $x_t$ at time $t$.  
A finite-state requirement --- which reflects the constraints of the real world ---
is therefore introduced for both encoder and decoder.  Within this class of schemes, Lempel and Ziv provide 
a (non-trivial) converse. 

Similarly, if one permits arbitrary encoders and decoders for the individual noise sequence channel
coding setting, the maximum communication rate of $\log |\mathcal{X}|$ may be easily attained
 for any noise sequence $z^\infty$.  One merely sets the
encoder's output to the difference between the source sequence and the noise sequence $M^\infty - z^\infty$.
The channel cancels out the noise, and the decoder needs merely read the channel output to 
obtain the message $M^\infty$.
To avoid such triviality, a 
finite-state encoder/decoder constraint analogous to that of 
Lempel and Ziv is introduced.  Within this class of schemes,
a non-trivial converse is proven.

\item In traditional Shannon theoretic results, converses that apply for a particular
source/channel model are accompanied by achievability schemes that function
for that particular model.  For instance, given a BSC with crossover probability $p$,
Shannon's channel coding converse tells us no reliable sequence of coding schemes
can achieve a rate better than $1-h_b(p)$, and Shannon provides a rate $(1-h_b(p))$-achieving sequence of
codebooks customized for the BSC$(p)$.
One might similarly ask Lempel and Ziv that to accompany their converse for a particular source
sequence, they provide an achievability scheme
for that particular sequence.  The achievability they provide, however, goes several steps further:
it functions for \emph{any} individual source sequence.
Note that while this scheme that is suggested is infinite-state, it is also shown that by truncation
and repetition, achievability is also possible through a sequence of finite-state source-coding schemes.  

Similarly, both the infinite-state scheme of Lomnitz and Feder \cite{LomnitzF2011}
and the sequence of finite-state achievability schemes $\overline{\mathcal{F}^m}$ presented in Sec. \ref{sec:achievability}
achieve the channel-coding converses for *any* noise sequence.

\item
The \emph{compressibility} $\rho(x^\infty)$ of a sequence, introduced by Lempel and Ziv, is its best possible compression rate.
The analogous quantity in the channel-coding case is here referred to as 
the channel's \emph{porosity} $\sigma(z^\infty)$.  In the binary additive
 noise case it is demonstrated to be equal to  $1-\rho(z^\infty)$.
  Both are analogues of probabilistic quantities --- entropy rate
 in the case of compressibility, and one minus the entropy rate in the case of porosity.
 Both may also be interpreted as generalizations of their probabilistic analogues, 
 since according to Thm. 4 in \cite{ZivL1978}, $\rho(X^\infty) = H(\mathbb{X})$ with probability 1
 if $X^\infty$ is an ergodic source.
\end{itemize}

\noindent To summarize, we show both that the porosity of the noise is the best possible rate achievable
within the class of finite-state schemes, and that there exist a sequence of finite-state schemes that simultaneously achieves
porosity for all noise sequences.

\section{Related Work}
\label{sec:related-work}
Lomnitz and Feder introduce the notion of competitive universality to channel coding in \cite{LomnitzF2011}.
The reference class used in this comparison consists of 
\emph{iterated fixed-blocklength} (IFB) schemes, which ignore the feedback channel
and simply employ block coding across the noisy channel.  
\emph{Rate-adaptive} schemes, on the other hand, make arbitrary use of feedback and communicate a fixed
number of bits over at most $n$ channel uses.  It is proven that
IFB schemes can do no better than porosity (rate $1-\overline{\rho}(z^\infty)$),
and a rate-adaptive scheme built upon LZ78 is shown to achieve porosity.

In a sense, the results reported here take these statements of competitive optimality
a step further: we ask that the achievability schemes not only outperform any elements of the reference class,
but that they \emph{are} elements of the reference class.  This establishes porosity as a channel capacity of sorts.
As IFB codes 
 frequently cannot even
achieve porosity for a given noise sequence $z^\infty$, let alone for the entire set of noise sequences,
 this requires that the reference class
be widened to the class of all finite-state schemes.

The porosity-achieving rate-adaptive scheme introduced by Lomnitz and Feder does not quite fall into this class,
as it consists of infinite states.  One might consider consider truncating and repeating it
in order to construct a finite state scheme.  While this could potentially work, we find it somewhat easier to build from the schemes of Shayevitz and Feder \cite{ShayevitzF2009}, whose performance guarantees mesh well with the asymptotic performance metrics
of interest here.

In \cite{ShayevitzF2009}, Shayevitz and Feder establish the initial results 
that have sparked much of the subsequent work in this problem.  An extremely general family
of channels is considered, but the results provided are most meaningful
when restricting attention to the individual additive noise sequence setting.
Of interest are two things: the construction of variable-rate, fixed-blocklength
schemes expanded from Horner's coding method, and the strong performance
guarantees that are provided.  The schemes are shown to achieve for any noise sequence the 
\emph{empirical capacity}, or one minus the first-order empirical entropy.
By operating this coding technique over \emph{blocks} of channel use, one can
potentially generalize to arbitrary-order empirical entropies.
The achievability schemes $\{\overline{\F^m}\}$ presented in this paper are
an extension of this idea.

As with both \cite{LomnitzF2011} and \cite{ShayevitzF2009}, Eswaran et al. \cite{EswaranSSG2010}
consider a very broad class of channels with noiseless feedback, but the results provided
are most meaningful in the modulo-additive setting with an individual noise sequence.
Extending  \cite{ShayevitzF2009}, it is demonstrated that even when the feedback 
is asymptotically zero-rate, the empirical capacity is still universally achievable.

As previously mentioned, even in the absence of a feedback link 
questions of universal channel coding can be considered.  The principal complication in this setting
is that the encoder no longer has any information about the specific channel,
and so neither the rate nor the transmission methodology can be adapted.

One may nonetheless ask that the \emph{decoder} adapt to the channel.  
  Csiszar and Korner \cite{CsiszarK2011} consider
the class of memoryless DMC's that share a common input alphabet $\mathcal{X}$ --- call this class $A(\mathcal{X})$.  
For a randomly generated codebook, a universal decoder for the entire class $A(\mathcal{X})$ is constructed 
and its performance compared to that of a decoder customized for whichever specific channel happens to appear
(maximum likelihood decoder).
The universal decoder is not only shown to match the ML decoder 
in terms of vanishing error, but it is also found to achieve the same error exponent.

Ziv \cite{Ziv1985} and Lapidoth and Ziv \cite{LapidothZ1998} seek to expand such a result into the territory
of channels with memory.  Each considers a fairly specific form of memory: finite-state
channels with deterministic state transitions \cite{Ziv1985} and those with probabilistic
state transitions \cite{LapidothZ1998}.  Each also demonstrates achievability through decoding schemes
that utilize LZ78-style sequence parsing.  Feder and Lapidoth \cite{FederL1998} on the other
hand consider the more general problem of decoding for a parametric family of channels.

Despite the non-adaptability of the encoder in this feedback-less setting, one may seek to
maximize the worst-case rate of communication across the channel.  The fundamental limit
of performance in this scenario is the \emph{compound channel capacity}, discussed at length
in the review article by Lapidoth and Narayan \cite{LapidothN1998}.

A generalization of the above is to take a broadcast approach,
wherein channel uncertainty is modeled by having the encoder
broadcast across all channels in the class considered.  The
rate region of this broadcast channel then characterizes the rates
the encoder may simultaneously achieve for each potential channel.
Observe that if this rate region can be specifically determined,
it answers all possible questions of universal decoding.
Shamai and Steiner \cite{ShamaiS2003} leverage this approach for
the case of fading channels.

\section{Structure of Paper}
In Sec. \ref{sec:ProblemSetup} a precise description is provided  of the
problem setting, the class of finite-state schemes, and the relevant performance metrics.
Sec. \ref{sec:NotionsOfCompressibility} builds slightly on Lempel and Ziv's definition of 
compressibility and establishes certain useful properties.
Sec. \ref{sec:StatementOfResults} states the three theorems that constitute the core results of this work.
Sec. \ref{sec:converse} contains the proof of the converse theorems, and Sec. \ref{sec:achievability}
proves achievability.  
Sec. \ref{sec:predictability} introduces a significantly more practical (but sub-optimal) set of schemes that establish a 
connection between porosity and finite-state predictability.
Sec. \ref{sec:Conclusion} summarizes this paper's findings.
A few lemmas have somewhat distracting proofs that are relegated to the appendices.

\section{Problem setup}
\label{sec:ProblemSetup}
A deterministic additive noise feedback channel, as depicted in
Fig. \ref{fig:channel-coding-additive}, is defined by  a noise sequence
$z^\infty \in \mathcal{X}^\infty$, where $\mathcal{X}$ is a finite alphabet with a
modulo-addition operator.  The channel output at any time $i$ is given by the sum
of the noise and the input: $y_i = x_i + z_i$.  
Noiseless feedback $u_i = y_{i-1}$ delays the channel output by one time unit
before providing it to the encoder.
Without loss of generality, we will concern ourselves primarily with the binary-alphabet case, i.e. $\mathcal{X} = \{0,1\}$,
as the extension to general finite $\mathcal{X}$ is straightforward.

\subsection{Finite-state Schemes}

\begin{figure}
 \centering
  \psfrag{c}[cc][cc]{\small Unknown Channel}
  \psfrag{m}[cc][cc]{$\displaystyle M_1,\ldots,\underbrace{M_{p_i},M_{p_i+1},\ldots,M_{p_i+\ell}},M_{p_i+\ell+1},\ldots$}
  \psfrag{e}[cc][cc]{E}
  \psfrag{d}[cc][cc]{D}
  \psfrag{m2}[cc][cc]{$\widehat{M}_{p_i}^{p_i+L_i-1}$}
  \psfrag{x}[cc][cc]{$x_i$}
  \psfrag{y}[cc][cc]{$y_i$}
  \psfrag{z}[cc][cc]{$z^{-1}$}
  \psfrag{y-1}[cc][cc]{$y_{i-1}$}
  \psfrag{2}[cc][cc]{$2$}
  \psfrag{noise}[cc][cc]{$z^{\infty}$}
  \psfrag{+}[cc][cc]{$+$}
  \psfrag{t}[cc][cc]{$\theta_i$}
  \includegraphics[width=4.8in]{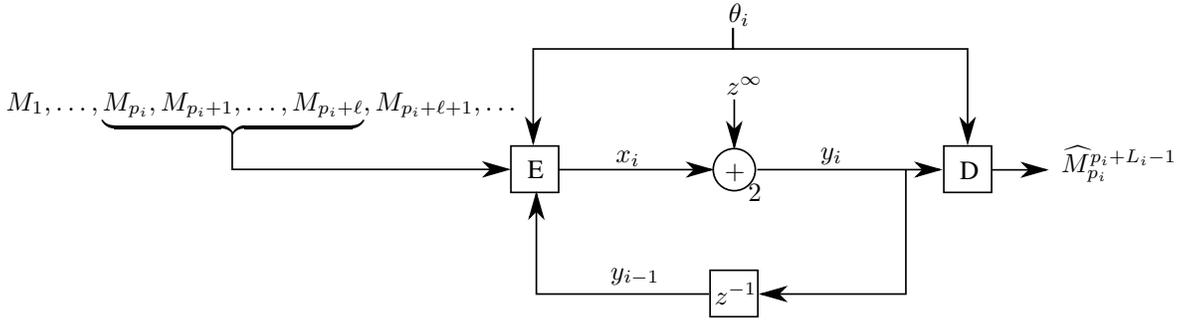}
  \caption{A finite-state encoding/decoding scheme for a 
  modulo-additive channel.}
 \label{fig:fs-schemes}
\end{figure}

A finite-state (FS) encoder/decoder  scheme for an additive noise channel ---
depicted in Fig. \ref{fig:fs-schemes} --- consists of several
components:
\begin{enumerate}
\item An encoder state variable $s_i^{\text{(e)}}$ and decoder state variable $s_i^{\text{(d)}}$,
each taking values in a finite set $S$.
\item A source pointer $p_i$ and a finite lookahead constant $\ell$.
\item An iid common randomness source $\theta_i \sim p_\theta$ taking values in a finite alphabet.
\item An encoding function $x_i = e(s_i^{\text{(e)}}, M_{p_i}^{p_i+\ell}, y_i, \theta_i) \in \mathcal{X}$.
\item A decoding length function $L_i = d_L(s_i^{\text{(d)}},y_i, \theta_i)$ that also determines the update of the source pointer:
$p_{i+1} = p_i + L_i$.
\item A decoding function $\widehat{M}_{p_i}^{p_i + L_i - 1} = d_M(s_i^{\text{(d)}},y_i, \theta_i)$.
\item State-update functions for both the encoder $s_{i+1}^{\text{(e)}} = f_{\text{(e)}}(s_i^{\text{(e)}},M_{p_i}^{p_i+\ell}, y_i, \theta_i)$
and decoder $s_{i+1}^{\text{(d)}} = f_{\text{(d)}}(s_i^{\text{(d)}}, y_i, \theta_i)$.
\end{enumerate}

At each time step, the encoding function determines the input $x_i$ to the channel,
the decoding function estimates the first $L_i$ source symbols that have yet to be estimated
(based on the output
$y_i$ of the channel), and state variables and the source pointer location are updated in anticipation
of the next transmission.

Observe, first, that this class of schemes is sufficiently general to include
the following as special cases:
\begin{enumerate}
\item The class of ``iterated fixed-length" block schemes, as defined by Lomnitz and Feder \cite{LomnitzF2011}.
These are simply block codes that ignore the feedback.  The common randomness at encoder and decoder allows for
randomly generated block codes as well.
\item Schemes that transmit a variable number of source symbols over a fixed number of channel uses, before
reseting their state variables and repeating the operation (defined more precisely in
Sec. \ref{sec:achievability-classes-of-schemes} as ``Repeated Finite-Extent'' schemes).
\item Schemes that transmit a variable (or fixed) number of source symbols over a variable (but bounded) number
of channel uses (also known as ``rate-adaptive'' schemes \cite{LomnitzF2011}).
\end{enumerate}

Secondly, notice that without certain restrictions in the definition of 
class FS, the problem can become trivial:
\begin{enumerate}
\item Suppose that the encoder is permitted to be infinite-state.  The system designer may then
allow the encoder state $s_i^{\text(e)}$ 
to be the current time index $i$.  This then allows the encoding function to be a function of $i$,
which in turn permits the encoding function to be customized for a particular noise sequence $z^\infty$:
 $e(i,M_{p_i}) = M_{p_i} - z_i$.  Sending this through the channel, $z_i$ is canceled out.  The 
 decoder needs merely read the channel output to obtain the message at the maximum possible rate, $\log | \mathcal{X} |$.
\item Suppose that the decoder is permitted to be infinite-state.  One may reverse the above construction
by having the encoder blindly send the message bits through the channel $e(M_{p_i}) = M_{p_i}$ 
and asking the decoder to cancel out the noise.  Specifically, letting $s_i^{\text{(d)}} = i$, the decoding
function can be a function of the time $i$.  This allows for a clever system designer to choose
$d_M(i,y_i) = y_i - z_i$, which guarantees that $\widehat{M}_{p_i} = M_{p_i}$ and that $L_i = 1$ for any $i$.
\item Finally, suppose that the finite-lookahead requirement is nonexistant --- that is, the encoding function
can look at the entire untransmitted message stream $x_i = e(s_i^{\text{(e)}}, M_{p_i}^\infty,y_i,\theta_i)$.
As we will illustrate, this is identical to allowing the encoder an infinite number of states.  If $M^\infty$ is 
a Bernoulli$(1/2)$ sequence, then with probability one there exists a one-to-one map between 
$M_i^{\infty}$ and $i$.  The encoder may therefore send $e(M_{p_i}^\infty) = M_{p_i} + z_{p_i}$ as the channel input at time $i$.  The decoder, as before, simply reads the channel output, achieving the maximum rate $\log| \mathcal{X} |$.
\end{enumerate}

\subsection{Performance metrics}
Channel coding typically concerns itself with the tradeoff between
rate of communication and the frequency of errors.  In the individual
sequence setting of interest to us, we define the instantaneous rate and
bit-error rate of an FS scheme at time $n$ as
\[
R_n = \frac{1}{n} \sum_{i=1}^n L_i \mbox{,}
\]
and 
\[
\epsilon_n = \frac{1}{n R_n} \sum_{i=1}^{nR_n} {\bf 1}_{\widehat{M}_i \neq M_i} \mbox{.}
\]
We consider two interpretations of these quantities.

{\bf Best-Case.} An FS scheme \emph{best-case $p$-achieves} rate/error $(R,\epsilon)$ for a noise sequence $z^\infty$ if
with at least probability $p$ there exists a sequence of points $\{n_i\} \in \Z^+$ such that $\lim_{i\rightarrow\infty} R_{n_i} \geq R$
and $\lim_{i\rightarrow\infty} \epsilon_{n_i} \leq \epsilon$.  In other words,
a performance monitor that observes the system at the ``right'' times 
will see it achieve $(R,\epsilon)$ with probability at least $p$.  If $p$ is 1, we say that the scheme
simply \emph{best-case achieves} $(R,\epsilon)$.

{\bf Worst-Case.}  An FS scheme \emph{worst-case $p$-achieves} rate/error $(R,\epsilon)$ if with 
at least probability $p$ both
$\liminf_{n\rightarrow\infty} R_n \geq R$ and $\limsup_{n\rightarrow\infty} \epsilon_n \leq \epsilon$.
In other words, a performance monitor observing the system at any set of sample times
will see it achieve $(R,\epsilon)$ with probability at least $p$.  If
$p$ is one the scheme is said to \emph{worst-case achieve} $(R,\epsilon)$.

Observe that the randomness in these definitions has two possible sources: 
the source sequence $M^\infty$
and the common-information sequence $\theta^\infty$ used by the FS scheme.
Sometimes the source $M^\infty$ will be a fixed sequence, but this is always made clear from context.

\section{Notions of ``compressibility''}
\label{sec:NotionsOfCompressibility}
The results of this paper connect the operational notions of achievability 
to certain long-established individual sequence properties, first introduced by
Lempel and Ziv \cite{ZivL1978} in a source-coding context.  In this section, 
these properties are defined and some useful relations are presented between them.

First, we denote the $k$th order block-by-block empirical distribution 
\[ 
\hat{p}^k(X^k)[x^n] = \frac{1}{\lfloor \frac{n}{k} \rfloor} \sum_{i=0}^{ \lfloor n/k \rfloor} 
                                                       {\bf 1}_{x_{ki+1}^{k(i+1)} = X^k}
\mbox{.}
\]
If the empirical distribution is instead computed in a 
sliding-window manner, we denote
\[
\hat{p}_{\text{sw}}^k(X^k)[x^n] = 
\frac{1}{n-k+1} \sum_{i=0}^{ n-k+1 } {\bf 1}_{x_{i+1}^{i+k} = X^k}
\mbox{.}
\]
The argument $[x^n]$ is occasionally omitted when the 
context is clear.

The $k$th order block-by-block empirical entropy is indicated
by $\hat{H}^k(x^n) = H_{\hat{p}^k}(X^k)$.  The sliding-window
$k$th order empirical entropy is similarly written as
$\hat{H}^k_{\text sw}(x^n) = H_{\hat{p}^k_{\text sw}}(X^k)$.

As shown by Ziv and Lempel \cite{ZivL1978}, the finite-state \emph{compressibility} of a sequence $x^\infty$
may be written as
\[
\overline{\rho}(x^\infty) = \limsup_{k\rightarrow\infty}\limsup_{n\rightarrow \infty}
\hat{H}^k_{\text sw}(x^n)
\mbox{.}
\]
An analagous quantity may also be introduced:
\[
\underline{\rho}(x^\infty) = \liminf_{k\rightarrow\infty}\liminf_{n\rightarrow \infty}
\hat{H}^k_{\text sw}(x^n)
\mbox{.}
\]
Operationally, compressibility is the smallest limit supremum compression ratio achievable
for a sequence (Theorem 3 in \cite{ZivL1978}).  It is not difficult to show that, analogously,
the second quantity is the smallest possible limit infimum compression ratio.  Informed by this,
we refer to the original compressibility quantity as the \emph{worst-case compressibility}
and the new limit infimum version as the \emph{best-case compressibility}.

The following lemma, proved in Appendix \ref{app:compressibility-proof}, demonstrates that both best-case and
worst-case compressibilities may be computed using either block-by-block or sliding-window
empirical entropies.

\begin{lemma}
\label{lem:compressibility}
Let $x^\infty$ be a finite-alphabet sequence.  Then
\[
\underline{\rho}(x^\infty)
= \lim_{k\rightarrow\infty} \liminf_{n\rightarrow\infty} \frac{1}{k} \hat{H}^k(x^n)
\]
and
\[
\overline{\rho}(x^\infty)
= \lim_{k\rightarrow\infty} \limsup_{n\rightarrow\infty} \frac{1}{k} \hat{H}^k(x^n) \mbox{.}
\]
\end{lemma}

The \emph{porosity} of a noise sequence $z^\infty \in \mathcal{Z}^\infty$ 
is defined in best-case 
\[ \overline{\sigma}(z^\infty) = \log_2 |\mathcal{Z}| - \underline{\rho}(z^\infty) \] and worst-case 
\[ \underline{\sigma}(z^\infty) = \log_2 |\mathcal{Z}| - \overline{\rho}(z^\infty)\] varieties as well.  
Observe the sign changes: while a ``good'' compressibility is small, a ``good'' porosity is large.
The remainder of this paper clarifies the operational significance of these quantities.

\section{Statement of results}
\label{sec:StatementOfResults}
The results of this paper may be summarized as follows:
\begin{enumerate}
\item A converse that upper-bounds the best-case achievable rate by an FS scheme.
\item A converse that upper-bounds the worst-case achievable rate by an FS scheme.
\item A sequence of universal FS schemes $\{\overline{\mathcal{F}^m}\}_{m=1}^\infty$ that simultaneously achieve the 
best-case and worst-case converse bounds for \emph{any} noise sequence $z^\infty$.
\end{enumerate}

\noindent Formally, each of these three statements corresponds to a theorem:

\begin{theorem}
\label{thm:best-case-converse}
Suppose an FS scheme best-case $p$-achieves $(R,\epsilon)$. 
If $p>0$, then 
\[
R \leq h_b(\epsilon) + \overline{\sigma}(z^\infty) \mbox{.}
\]
\end{theorem}

\begin{theorem}
\label{thm:worst-case-converse}
Suppose an FS scheme worst-case $p$-achieves $(R,\epsilon)$. 
If $p>0$, then 
\[
R \leq h_b(\epsilon) + \underline{\sigma}(z^\infty) \mbox{.}
\]
\end{theorem}

\begin{theorem}
\label{thm:achievability}
There exists a sequence of schemes $\overline{\mathcal{F}^m}$,
which for an iid Bernoulli$(1/2)$ source $M^\infty$ and every noise sequence $z^\infty$
best-case achieves
\[ \overline{\sigma}(z^\infty) - \underline{\delta}_m(z^\infty) \; , \; \epsilon_m/(\underline{\sigma}(z^\infty) - \overline{\delta}_m(z^\infty)) \mbox{,}\]
and worst-case achieves
\[
\underline{\sigma}(z^\infty) - \overline{\delta}_m(z^\infty) \; , \;
\epsilon_m/(\underline{\sigma}(z^\infty) - \overline{\delta}_m(z^\infty)) \mbox{,}
\]
with probability one, where $\epsilon_m$, $\underline{\delta}_m$, and $\overline{\delta}_m$
all go to zero.
\end{theorem}

Theorems \ref{thm:best-case-converse} and \ref{thm:worst-case-converse} are
proven in Section \ref{sec:converse}.  In Section \ref{sec:achievability},
we introduce the schemes $\{\overline{\mathcal{F}^m}\}_{m=1}^\infty$ and
prove Theorem \ref{thm:achievability}.

\section{Proof of Converse}
\label{sec:converse}
\subsection{Definitions and Lemmas}
In order to prove the converse theorems, a series of definitions
and lemmas is first required.  

\begin{lemma}
\label{lem:selection}
(Selection Lemma)
Suppose $X^\infty$ is iid Bernoulli$(1/2)$ and
$L$ is a random positive integer with arbitrary conditional distribution
$p_{L |X^\infty}$ with respect to $X^\infty$.  Then $H(X^L) \geq \E{L}$.
\end{lemma}
\begin{proof}
\begin{eqnarray*}
H(X^L) & = & \sum_{x^l} p(X^L = x^l) \log \frac{1}{p(X^L = x^l)} \\
& \geqlabel{a} & \sum_{x^l} p(X^L = x^l) \log \frac{1}{p(X^l = x^l)} \\
& \geqlabel{b} & \sum_{x^l} p(X^L = x^l)l \\
& = & \E{L}
\end{eqnarray*}
where step (a) follows because if $X^L = x^l$ then $X^l$ must
necessarily equal $x^l$.  Step (b) follows from the iid Bernoulli $(1/2)$
distribution of $X^\infty$.
\end{proof}

\begin{definition}
Let $\{L_i\}_{i=1}^{\infty}$ be a bounded sequence of nonnegative integers,
and let $M^\infty$ and $z^\infty$ as usual denote binary sequences.  
The $k$-partition of $(M^\infty,z^\infty)$ according to $\{L_i\}$
is the sequence of blocks
\[
(M^{L_i},z^k)_i = M_{\sum_{j=1}^{i-1}L_j + 1}^{ \sum_{j=1}^i L_j }, z_{(i-1)k+1}^{ik} \mbox{.}
\]
In this context, $\{L_i\}$ are referred to as the partition lengths.
\end{definition}

\begin{definition}
\label{def:limiting-distribution}
Let $x^\infty$ be a sequence of symbols drawn from a finite alphabet $\mathcal{X}$.
If there exists a series of sample points $\{n_i\}_{i=1}^\infty$ such that
the sequence $\hat{p}^1(x)[x^{n_i}]$ converges to a distribution $\hat{p}(x)$,
$\hat{p}(x)$ is said to be a limiting distribution for $x^\infty$.
\end{definition}
Observe that for any finite-alphabet sequence $x^\infty$ at least one limiting
distribution exists: $\hat{p}^1(x)[x^n]$ is an infinite sequence in a compact
set, so at least one convergent subsequence must exist.

\begin{definition}
Let $z^\infty$ be a finite-alphabet sequence.  The set $\mathcal{M}_k(z^\infty)$ 
consists of all binary sequences $M^\infty$ such that there exist partition lengths
$\{L_i\}$, a resulting $k$-partition $\{(M^{L_i},z^k)_i\}$,
and a limiting distribution $\hat{p}(L,M^L,z^k)$ for the sequence $\{L_i,(M^{L_i},z^k)_i\}$ such that
\beq
\E{L}_{\hat{p}} > H_{\hat{p}}(M^L | z^k) + 1 \mbox{.} \label{eq:m-definition}
\eeq
\end{definition}
We may interpret the set in the following manner.  Suppose first that a
``genie'' partitions the source sequence $M^\infty$ into an arbitrary series 
of variable-length blocks $\{(M^{L_i})_i\}$.  Each block $(M^{L_i})_i$, of length $L_i$, is
then source-coded with side information $(z^k)_i$ at average rate less than $H_{\hat{p}}(M^L | z^k) + 1$.  The set $\mathcal{M}_k(z^\infty)$
consists of all source sequences that, in a sense, allow such a genie/side-information-source-coding
setup to compress strictly better than one bit per source symbol.  One would expect that the occurrence of such a set is a rare
event when the source is drawn iid Bernoulli$(1/2)$.  This is formalized with the following lemma.

\begin{lemma}
\label{lem:limiting-distribution}
Let $z^\infty$ be a fixed finite-alphabet sequence, and let $M^\infty$
be drawn from an iid Bernoulli$(1/2)$ process.
Then the probability that $M^\infty \in \mathcal{M}_k(z^\infty)$
is zero for all $k$.
\end{lemma}
\begin{proof}
See Appendix \ref{app:limiting-distribution-proof}.
\end{proof}

We may easily expand this lemma to allow for common randomness.
\begin{corollary}
\label{cor:limiting-distribution}
Let $z^\infty$ be a fixed binary sequence, let $M^\infty$ 
be drawn from an iid Bernoulli$(1/2)$ process,
and let $\theta^\infty$ be a finite-alphabet sequence of arbitrary distribution that is independent of $M^\infty$.
Then the probability that $M^\infty \in \mathcal{M}_k((z_i,\theta_i)_{i=1}^\infty)$
is zero for any $k$.
\end{corollary}

\begin{proof}
First observe that Lemma \ref{lem:limiting-distribution} does not require
that $z^\infty$ be a binary sequence --- only that it be of a finite
alphabet.  As such, for a given sequence $\theta^\infty$ we may define the surrogate $z$-sequence 
$\widetilde{z}^\infty = (z_i,\theta_i)_{i=1}^\infty$.

Applying Lemma \ref{lem:limiting-distribution} for this surrogate sequence,
we have that for any fixed $z^\infty$ and $\theta^\infty$,  the probability of drawing an element of
$\mathcal{M}_k((z_i,\theta_i)_1^\infty)$ from a Bernoulli$(1/2)$ process is zero.  Since
$\theta^\infty$ is independent from $M^\infty$, the corollary follows.
\end{proof}

\subsection{Converse Lemma}
Although the converse results are presented as two distinct theorems,
at their heart is the same argument.  We present this core result
in the following lemma.

\begin{lemma}
\label{lem:converse}
Suppose an $s$-state $\ell$-lookahead FS scheme achieves $(R,\epsilon)$
 on points $\{n_i\}$ for a specific source sequence 
$M^\infty$, a specific channel noise sequence $z^\infty$, and a specific encoder/decoder common information
sequence $\theta^\infty$.  If for some $k \in \mathbb{Z}^+$ $M^\infty$ is not a member of $\mathcal{M}_k((z_i,\theta_i)_{i=1}^\infty)$,
and if $\hat{H}^k(z^n) + \hat{H}^k(\theta^n) - \hat{H}^k((z_i,\theta_i)_{i=1}^n) \rightarrow_{n\rightarrow\infty} 0$, then
\[
R \leq \frac{2 \log s + \ell + 2}{k} + h_b(\epsilon)
+ 1 - \limsup_{i\rightarrow\infty}\frac{1}{k} \widehat{H}^k(z^{n_i}) \mbox{.}
\]
\end{lemma}

The general idea in proving this lemma is to turn any given FS scheme into a \emph{source} encoding/decoding scheme.
Consider an FS decoder that achieves $(R,\epsilon)$ on some points $\{n_i\}$, and ignore the minor complication
of common randomness $\theta^\infty$.
Given only the channel output $y^\infty$, the decoder produces an estimate of the source sequence $M^\infty$.
Knowing the source sequence and the channel output, the decoder is technically capable of ``simulating'' the encoder
and thereby obtaining both the channel input sequence $x^\infty$ and the noise sequence $z^\infty$.
One may therefore interpret the channel output $y^\infty$ as an encoding of the joint source sequence $(M^\infty,z^\infty)$.
The following proof utilizes a rigorous argument inspired by this intuition.

\begin{proof}
Let $e^\infty$ denote the error indication sequence $e_i = \mathds{1}_{M_i \neq \widehat{M}_i}$.
First, consider the $k$-partition of $(M^\infty,z^\infty,\theta^\infty, e^\infty)$
according to the given FS scheme:
\[
(M^{L_i},e^{L_i}, z^k,\theta^k)_{i=1}^\infty = \left( M_{p_{(i-1)k+1}}^{p_{ik+1}-1}, e_{p_{(i-1)k+1}}^{p_{ik+1}-1},  z_{(i-1)k+1}^{ik},\theta_{(i-1)k+1}^{ik} \right)_{i=1}^\infty \mbox{.}
\]
In other words, let $(M^{L_i},e^{L_i}, z^k,\theta^k)_i$ enumerate $k$-blocks of channel noise and common information, along with the source bits
that are estimated during each such block and the error indicators for these source bits.  Let $L_i = p_{ik+1} - p_{(i-1)k+1}$ denote the partition lengths.

Now define the sequence of points $\{n_i^*\} \subset \{n_i\}$ so that
\beq
\lim_{i\rightarrow\infty} \frac{1}{k} \widehat{H}^k(z^{n_i^*}) =
\limsup_{i\rightarrow\infty}\frac{1}{k}\widehat{H}^k(z^{n_i}) \mbox{,}
\label{eq:worstCaseLimit}
\eeq
and let $p(M^L,e^L,z^k,\theta^k)$ be a limiting distribution
of $(M^L,z^k,\theta^k)_i$ on these points $\{n_i^*\}$.  Recall
from Def. \ref{def:limiting-distribution} that such a limiting distribution always exists.

Suppose random variables $(M^L,e^L,z^k,\theta^k)$ are distributed according to $p(M^L,z^k,\theta^k)$.
We first use the FS scheme given in the lemma statement to construct a lossless source encoder/decoder for $(M^L,z^k)$, with
$\theta^k$ as side information.  By later requiring that the rate of this encoding exceed $H_p(M^L,z^k | \theta^k)$, the lemma may be proven.


\begin{description}
\item[E1]  Let $j\in \Z^+$.  We construct
a codeword for the source block $(M^{L_j},z^k)_j$
given side information block $(\theta^k)_j$ as follows:

\begin{enumerate}
\item To reduce clutter, we remove some of the unnecessary indices.
Denote the source bits used by the FS encoder during this $j$th block as 
$M^{\ell} = M_{p_{(j-1)k+1}}^{p_{(j-1)k +1} + \ell-1}$.
Similarly, let the source bits estimated by the FS decoder be referred to as 
$\widehat{M}^L = \widehat{M}_{p_{(j-1)k+1}}^{p_{jk+1}-1}$ and
the error indicators as $e^L = e_{p_{(j-1)k+1}}^{p_{jk+1}-1}$.
Additionally, let $s^{\text{(e)}} = s_{(j-1)k+1}^{\text{(e)}}$ and 
$s^\text{(d)} = s_{(j-1)k+1}^{\text{(d)}}$ indicate the 
initial encoder and decoder states, and let
$x^k = x_{(j-1)k+1}^{jk}$ and $y^k = y_{(j-1)k+1}^{jk}$
denote the channel inputs and outputs during the block.

%
%
%
\item Apply a binary Huffman code to $e^L$ to create the compressed representation $g(e^L) \in \{0,1\}^*$.
The Huffman code is designed according to the limiting empirical distribution $p(e^L)$.

\item Add the codeword $(s^{\text{(e)}}, s^{\text{(d)}}, M_{L+1}^{L+\ell}, y^k, g(e^L))$
to the codebook.

\item Observe that $(s^{\text{(e)}}, s^{\text{(d)}}, M_{L+1}^{L+\ell}, g(e^L))$
decodes uniquely into $(M^L,z^k)_j$ given side information block $(\theta^k)_j$:
\begin{itemize}
	\item Simulate the channel decoding operation with 
	initial state $s^{\text{(d)}}$, common information $\theta^k$, and channel output
	$y^k$.  This yields
	$\widehat{M}^L$.  Correcting for errors with the correctional
	information embedded in $g(e^L)$, we have $M^L$.
	\item Simulate the channel encoding operation using initial
	state $s^{\text{e}}$, feedback $y^k$, 
	common information $\theta^k$, source
	$M^L$ from the previous step, and $M_{L+1}^{L+\ell}$ from the codeword.
	This yields the channel input $x^k$, which produces $z^k$ when modulo-2 added
	to $y^k$.
\end{itemize}
Refer to this decoding operation on a codeword $c$ as $\mathcal{C}^{-1}(c,\theta^k)$.
\end{enumerate}

\item[E2]  Build the codebook by repeating step {\bf E1} for every
block $(M^L,z^k,\theta_k)_j$, $j\in \Z^+$.  Note that each codeword is of
length $2 \log s + \ell + k + \text{length}(g(e^L))$ and losslessly decodes into its
source block.  Call this codebook $\widetilde{\mathcal{C}}$. 

\item[E3] Define the codebook encoding function $F$
of a source sample $(M^L,z^k | \theta^k)$ as mapping to the shortest codeword in the set
$\{c \in \widetilde{\mathcal{C}}: \widetilde{\mathcal{C}}^{-1}(c,\theta^k) = (M^L,z^k) \}$.
\end{description}

We now establish the expected length of this code when applied to the probabilistic
source $(M^L,z^k | \theta^k)$.  As mentioned in step E2, a given codeword is
of length $2 \log s + \ell + k + \text{length}(g(e^L))$, where $e^L$ is the error sequence.
According to the assumptions of the lemma, the bit error rate on points $\{n_i^*\}$
is upper-bounded by $\epsilon$.  From this and from $p(e^L)$ being a limiting
distribution on $\{n_i^*\}$, the expected frequency of $1$ in $e^L$ may be upper-bounded
by $\epsilon$.  Therefore, the expected length of $g(e^L)$ is upper-bounded by $k h_b(\epsilon) + 1$,
and
\[ 
\E{\text{length}(F(M^L,z^k | \theta^k))}_p \leq 2 \log s + \ell + k(1+h_b(\epsilon))+1 \mbox{.}
\]


Since $F$ is a lossless variable-length encoder for sources drawn from
$p(M^L,z^k|\theta^k)$, the expected codeword length must exceed
the conditional entropy according to $p$:
\begin{eqnarray*}
2 \log s + \ell + k\left(1+h_b(\epsilon) \right) + 1&\geq & H_p(M^L,z^k|\theta^k) \\
&\geq & H_p(z^k|\theta^k) + H_p(M^L|z^k,\theta^k) \\
& \eqlabel{a} & H_p(z^k) + H_p(M^L|z^k,\theta^k) \\
&\geqlabel{b} & H_p(z^k) + \E{L}_p - 1 \\
& \eqlabel{c} & \limsup_{i\rightarrow\infty} \widehat{H}^k(z^{n_i}) + Rk - 1
\end{eqnarray*}
where (a) holds because of the final assumption in the theorem statement, (b) follows from Corollary \ref{cor:limiting-distribution},
and (c) is due to both \eqref{eq:worstCaseLimit} and the lemma's assumption about rate $R$ being achieved on points
$\{n_i^*\}$.  Rearranging terms proves the lemma.
\end{proof}

\subsection{Proof of Converses}
Armed with Lemma \ref{lem:converse}, it is a relatively straightforward matter
to prove Theorems \ref{thm:best-case-converse} and \ref{thm:worst-case-converse}.

\subsubsection*{Proof of Theorem \ref{thm:best-case-converse}}
We first note that because $\theta^\infty$ is drawn iid and $z^\infty$ is fixed,
 $\hat{H}^k(z^n) + \hat{H}^k(\theta^n) - \hat{H}^k((z_i,\theta_i)_{i=1}^n) \rightarrow 0$ 
with probability one for every $k$.
Furthermore, by Corollary \ref{cor:limiting-distribution}, $M^\infty \notin \mathcal{M}_k((z_i,\theta_i)_{i=1}^\infty)$
with probability one for every k.  Therefore, if $(R,\epsilon)$ is best-case-achieved with positive probability,
it must then be achieved for some specific
 $(M^\infty,\theta^\infty)$ such that $M^\infty \notin \mathcal{M}_k((z_i,\theta_i)_{i=1}^\infty$
 and $\hat{H}^k(z^n) + \hat{H}^k(\theta^n) - \hat{H}^k((z_i,\theta_i)_{i=1}^n) \rightarrow 0$
 for every $k$.
Let $\{n_i\}$ be the subsequence on which it is achieved.

Applying Lemma \ref{lem:converse},
\[
R \leq \frac{2 \log s + \ell + \sqrt{k}}{k} + h_b(\epsilon) 
+ 1 - \frac{1}{k} \limsup_{i\rightarrow\infty}\widehat{H}^k(z^{n_i}) \mbox{,}
\]
for any $k$. 

Taking the limit supremum as $k\rightarrow\infty$,
\begin{eqnarray*}
R & \leq & h_b(\epsilon) 
+ 1 - \liminf_{k\rightarrow\infty} \frac{1}{k} \limsup_{i\rightarrow\infty}\widehat{H}^k(z^{n_i}) \\
& \leq & h_b(\epsilon) 
+ 1 - \liminf_{k\rightarrow\infty} \frac{1}{k} \liminf_{i\rightarrow\infty}\widehat{H}^k(z^{n_i}) \\
& \leq & h_b(\epsilon) 
+ 1 - \liminf_{k\rightarrow\infty} \frac{1}{k} \liminf_{n\rightarrow\infty}\widehat{H}^k(z^{n}) \\
& = &h_b(\epsilon) + 1 - \underline{\rho}(z^\infty) \mbox{.}
\end{eqnarray*}

\subsubsection*{Proof of Theorem \ref{thm:worst-case-converse}}
If $(R,\epsilon)$ is worst-case-achieved with positive probability $p$,
then it must be worst-case achieved for some $(M^\infty,\theta^\infty)$
such that $M^\infty \notin \mathcal{M}_k((z_i,\theta_i)_{i=1}^\infty$ (because by Corollary \ref{cor:limiting-distribution} this occurs with probability one)
and $\hat{H}^k(z^n) + \hat{H}^k(\theta^n) - \hat{H}^k((z_i,\theta_i)_{i=1}^n) \rightarrow 0$ 
(because this occurs with probability one when $\theta^\infty$ is chosen iid).

We may therefore apply Lemma \ref{lem:converse} with $\{n_i\} = \mathbb{Z}^+$.  For
any $k$, we have that
\[
R \leq \frac{2 \log s + \ell + \sqrt{k}}{k} + h_b(\epsilon) 
+ 1 - \frac{1}{k} \limsup_{n\rightarrow\infty}\widehat{H}^k(z^{n}) \mbox{.}
\]
Since this holds for arbitrary $k$, we may take the limit infimum of the expression
with $k\rightarrow \infty$:
\begin{eqnarray*}
R & \leq & \liminf_{k\rightarrow\infty} \left( \frac{2 \log s + \ell + \sqrt{k}}{k} + h_b(\epsilon) 
+ 1 - \frac{1}{k} \limsup_{n\rightarrow\infty}\widehat{H}^k(z^{n}) \right) \\
& = & h_b(\epsilon) + 1 - \limsup_{k\rightarrow\infty}\frac{1}{k} \limsup_{n\rightarrow\infty}\widehat{H}^k(z^{n}) \\
& = & h_b(\epsilon) + 1 - \overline{\rho}(z^\infty) \mbox{.}
\end{eqnarray*}


\section{Proof of Achievability}
\label{sec:achievability}
In this section, a sequence of FS schemes is constructed
and guaranteed to achieve both the best-case
and worst-case bounds for any channel noise sequence
$z^\infty$.  This ``universal achievability'' is
analogous to the universal source coding achievability scheme
introduced by Ziv and Lempel \cite{ZivL1978}.

\subsection{Some Classes of Schemes}
\label{sec:achievability-classes-of-schemes}
To begin, several additional classes of schemes are introduced,
and their relationships with each other and with class FS are clarified.
This will prove useful in constructing the universal achievability schemes.

A \emph{finite-state active feedback} (FSAF) scheme
is a variation of the class FS
that allows for active feedback (Fig. \ref{fig:fs-schemes}).  It consists of the following:
\begin{enumerate}
\item An encoder state variable $s_i^{\text{(e)}}$, a decoder state variable $s_i^{\text{(d)}}$,
and a feedback state variable $s_i^{\text{(f)}}$, all taking values in a finite set.
\item A source pointer $p_i$ and a finite lookahead constant $\ell$.
\item An iid common randomness source $\theta_i \sim p_\theta$.
\item A finite-state feedback channel whose output at time $i$ is distributed
according to $U_i \sim p_{u|y,s}(u_i | y_i, s_i^{\text{(f)}})$.
\item An encoding function $x_i = e(s_i^{\text{(e)}}, M_{p_i}^{p_i+\ell}, \theta_i, u_i)$.
\item A decoding length function $L_i = d_L(s_i^{\text{(d)}},y_i, \theta_i, u_i)$ that also determines the update of the source pointer:
$p_{i+1} = p_i + L_i$.
\item A decoding function $\widehat{M}_{p_i}^{p_i + L_i - 1} = d_M(s_i^{\text{(d)}},y_i, \theta_i, u_i)$.
\item State-update functions for  the encoder 
$s_{i+1}^{\text{(e)}} = f_{\text{(e)}}(s_i^{\text{(e)}},M_{p_i}^{p_i+\ell}, \theta_i, u_i)$,
 decoder $s_{i+1}^{\text{(d)}} = f_{\text{(d)}}(s_i^{\text{(d)}}, y_i, \theta_i, u_i)$,
 and feedback channel $s_{i+1}^{\text{(f)}} = f_{\text{(f)}}(s_{i}^{\text{(f)}}, y_i, \theta_i)$.
\end{enumerate}

\begin{figure}
 \centering
  \psfrag{c}[cc][cc]{\small Unknown Channel}
  \psfrag{m}[cc][cc]{$\displaystyle M_1,\ldots,\underbrace{M_{p_i},M_{p_i+1},\ldots,M_{p_i+\ell}},M_{p_i+\ell+1},\ldots$}
  \psfrag{e}[cc][cc]{E}
  \psfrag{d}[cc][cc]{D}
  \psfrag{m2}[cc][cc]{$\widehat{M}_{p_i}^{p_i+L_i-1}$}
  \psfrag{x}[cc][cc]{$x_i$}
  \psfrag{y}[cc][cc]{$y_i$}
  \psfrag{z}[cc][cc]{$\displaystyle p_{u_i|y_{i-1},s^{\text{(f)}}}$}
  \psfrag{y-1}[cc][cc]{$u_i$}
  \psfrag{2}[cc][cc]{$2$}
  \psfrag{noise}[cc][cc]{$z^{\infty}$}
  \psfrag{+}[cc][cc]{$+$}
  \psfrag{t}[cc][cc]{$\theta_i$}
  \includegraphics[width=4.8in]{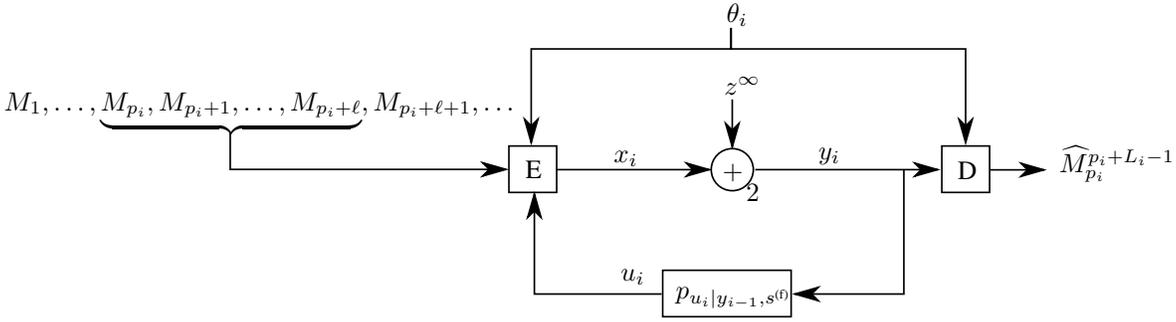}
  \caption{A finite-state active feedback scheme for a modulo-additive channel.}
 \label{fig:fsaf-schemes}
\end{figure}

\begin{lemma}
\label{lem:FSAFSisFSS}
The class of schemes FSAF is equivalent to class FS.
\end{lemma}
\begin{proof}
By setting $U_i = y_i$, we find that FS is a special case of FSAF.
To show the other direction, first assume we are given an FSAF scheme.
By the arguments that follow,
we will construct an FS scheme that simulates it.  
Quantities relating to the constructed $FS$ scheme will be notated with
a ``hat,'' e.g. $\hat{\theta},\hat{s}_i^{\text{e}}$, etc.

First, let ${\bf U}$ be a random vector with components
indexed by $y \in \mathcal{X}$ and $s \in \mathcal{S}^{\text{(f)}}$,
and let the $(y,s)$th component be distributed according to 
$U_{y,s} \sim p_{U|y,s}$, where $p_{U|y,s}$ is the given FSAF scheme's
feedback channel transition matrix.  Furthermore, let the components
be independently distributed.  We then define the
FS scheme's common randomness as $\hat{\theta}_i = (\theta_i,{\bf U}_i)$
distributed iid according to $p_\theta p_{U}$, where $p_\theta$ is
the common randomness distribution for the given FSAF scheme.
Let $(U_{y,s})_i$ denote the $(y,s)-$th component of the $i$-th
random vector ${\bf U}_i$.

The equivalent FS scheme is assigned encoder and decoder state
variables $\hat{s}^{\text(e)}_i = (s^{\text(e)}_i, s^{\text(f)}_i)$
and $\hat{s}^{\text(d)}_i = (s^{\text(d)}_i, s^{\text(f)}_i)$, where
$s^{\text(e)}_i$, $s^{\text(d)}_i$, and $s^{\text(f)}_i$ are 
the state variables for the FSAF scheme.  The encoder/decoder update
functions are given by
\[
\hat{f}_{\text{(e)}} (\hat{s}_i^{\text{(e)}},
											M_{p_i}^{p_i+\ell},
											y_i,
											\hat{\theta}_i)
 = 
\left( f_{\text{(e)}}( s_i^{(e)},
											 M_{p_i}^{p_i+\ell}, 
											 \theta_i, 
											 (U_{y_i,s_i^{\text{(f)}}})_i
										  ),
f_{\text{f}} (s_i^{\text{(f)}},
							y_i,
						  \theta_i)
\right) \mbox{,}
\]
and
\[
\hat{f}_{\text{(d)}} (\hat{s}_i^{\text{(d)}},
											y_i,
											\hat{\theta}_i)
 = 
\left( f_{\text{(d)}}( s_i^{(e)},
											 y_i, 
											 \theta_i, 
											 (U_{y_i,s_i^{\text{(f)}}})_i
										  ),
f_{\text{f}} (s_i^{\text{(f)}},
						  y_i,
						  \theta_i)
\right) \mbox{.}
\]
Observe how the randomness of the FSAF scheme's active feedback channel
is simulated by means of the common randomness $\hat{\theta}_i = (\theta_i,{\bf U}_i)$.

In this manner the FSAF encoder/decoder/feedback state machines
are simulated by the FS encoder/decoder state machines.  The FSAF
encoding and decoding functions may be implemented in a similar manner:
\[
\hat{e}(\hat{s}_i^{\text{(e)}},
  M_{p_i,p_i+\ell},
  y_i,
  \hat{\theta}_i) = 
e(s_i^{\text{(e)}},
  M_{p_i,p_i+\ell},
  \theta_i,
  (U_{y_i,s_i^{\text{(f)}}})_i)
\mbox{,}\]
and
\[
\hat{d}_M(\hat{s}_i^{\text{(d)}},
				y_i,
				\hat{\theta}_i) =
d_M(s_i^{\text{(d)}},
  y_i,
  \theta_i,
  (U_{y_i,s_i^{\text{(f)}}})_i)
\mbox{.}
\]
Since this constructed FS scheme is identical to the given
FSAF scheme, FSAF is a special case of FS. 
\end{proof}

A \emph{finite-extent} (FEex) scheme $\mathcal{F}$ for a channel
with alphabet $\mathcal{X}$ --- as depicted in Fig. \ref{fig:fh-schemes} consists of:

\begin{enumerate}
\item A extent $n$.
\item A feedback channel with transition probabilities given by
$U_i \sim p_{u_i}(u_i | u^{i-1},y^i)$ and taking values in $\mathcal{X}$, for
$i\in\{1,\ldots,n\}$.
\item A common randomness variable $\theta$ drawn from a finite alphabet, independent of the source, and provided to both encoder and decoder.
\item Encoding functions $x_1= e_1(M^\infty, \theta),x_2 =e_2(M^\infty,\theta, u_1),\ldots,x_n = e_n(M^\infty,\theta,u^{n-1})$.
\item A decoding length function $L = d_L(y^n,\theta, u^{n-1})$, upper bounded by 
$n \log |\mathcal{X}|$.
\item A decoding function $\widehat{M}^L = d_M(y^n,\theta,u^{n-1})$.
\end{enumerate}

\begin{figure}
 \centering
  \psfrag{c}[cc][cc]{\small Unknown Channel}
  \psfrag{m}[cc][cc]{$M^\infty$}
  \psfrag{e}[cc][cc]{E}
  \psfrag{d}[cc][cc]{D}
  \psfrag{m2}[cc][cc]{$\widehat{M}^L$}
  \psfrag{x}[cc][cc]{$x_i$}
  \psfrag{y}[cc][cc]{$y_i$}
  \psfrag{z}[cc][cc]{$\displaystyle p_{u_i|y^{i-1}}$}
  \psfrag{y-1}[cc][cc]{$u_i$}
  \psfrag{2}[cc][cc]{$2$}
  \psfrag{noise}[cc][cc]{$z^{\infty}$}
  \psfrag{+}[cc][cc]{$+$}
  \psfrag{t}[cc][cc]{$\theta$}
  \includegraphics[width=3.8in]{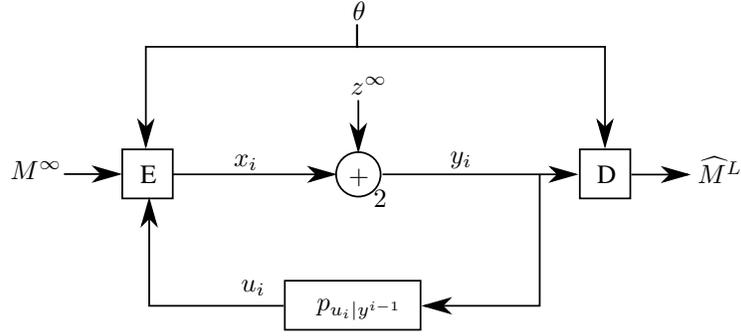}
  \caption{A finite-extent scheme.}
 \label{fig:fh-schemes}
\end{figure}

A \emph{repetition scheme} is constructed from an 
$n$-extent FE scheme $\mathcal{F}$.
Let $\mathcal{F}(M^\infty,z^n)$ describe the application
of scheme $\mathcal{F}$ to source $M^\infty$ and noise block
$z^n$.  Then the repetition scheme $\overline{\mathcal{F}}$
consists of repeated independent uses of $\mathcal{F}$, i.e.
\[
\overline{\mathcal{F}}(M^\infty,z^\infty) \equiv
\left\{ \mathcal{F}\left( (M_1^n,0^\infty),
													z_1^n 
									 \right),
				\mathcal{F}\left( (M_{L_1+1}^{L_1+n},0^\infty),
													z_{n+1}^{2n}
									 \right),
				\ldots
\right\}
\mbox{.} \]
In each block, $\mathcal{F}$ is applied to a ``virtual source''
consisting of the first $n$
bits of the source that have yet to be transmitted and a string of
0s.

\begin{proposition}
The class of repetition schemes is a subclass of FSAF schemes (and therefore
of FS schemes).
\end{proposition}
This follows directly from two properties of repetition schemes:
\begin{itemize}
\item The block-based structure allows for implementation with finite-state
machines.
\item A repetition scheme constructed from an $n$-extent FE scheme
has finite lookahead constant $n$.
\end{itemize}

\subsection{Universal Scheme Construction}
At the core of the achievability scheme is a lemma, 
introduced and proven by Shayevitz and Feder \cite{ShayevitzF2009}:

\begin{lemma}
\label{lem:SF2009}
(Shayevitz and Feder, 2009)
Let $\mathcal{X}$ be a finite alphabet with an addition operation.
Then there exists a sequence of $n$-extent FE schemes $\mathcal{F}_n(\mathcal{X})$
with the following worst-case performance guarantees over all additive noise
sequences $z^\infty \in \mathcal{X}^\infty$ and source sequences $M^\infty \in \{0,1\}^\infty$
\beq
\sup_{z^n \in \mathcal{X}^n, M^\infty \in \{0,1\}^\infty} \P{\widehat{M}^L \neq M^L} \leq \epsilon(n)
\label{eq:errorGuarantee}
\eeq
and
\beq
\inf_{z^n \in \mathcal{X}^n, M^\infty \in \{0,1\}^\infty} 
\P{ \frac{L}{n} > 1 - \widehat{H}^1(z^n)-\epsilon(n)} > 1 - \epsilon(n)
\mbox{,}
\label{eq:rateGuarantee}
\eeq
where $\epsilon(n) \rightarrow 0$.
\end{lemma}
Note that the only randomness in the above probabilistic statements is due
to the randomness in the feedback channel.

Observe that Lemma \ref{lem:SF2009} concerns itself with only the
first-order empirical entropy $\widehat{H}^1(z^n)$.  By specializing
to binary sequences, this may be replaced by higher-order empirical entropies.  
\begin{corollary}
\label{cor:FEConstruction}
For binary additive noise channels with feedback, there exists a sequence
of finite extent schemes $\mathcal{F}^m$ with extents $N(m)\rightarrow \infty$
and the following performance
guarantees:
\[
\sup_{z^{N(m)} \in \mathcal{X}^{N(m)}, M^\infty \in \{0,1\}^\infty} \P{\widehat{M}^L \neq M^L} \leq \epsilon_m
\]
and
\[
\inf_{z^{N(m)} \in \mathcal{X}^{N(m)}, M^\infty \in \{0,1\}^\infty} 
\P{ \frac{L}{N(m)} > 1 - \frac{1}{m}\widehat{H}^m(z^{N(m)})-\epsilon_m} > 1 - \epsilon_m
\mbox{,}
\]
where $\epsilon_m \rightarrow 0$.
\end{corollary}

\begin{proof}
For a given $m$, consider the $m$-tuple supersymbol channel characterized
by inputs $X_i = x_{(i-1)m+1}^{im}$, noise $Z_i = z_{(i-1)m+1}^{im}$,
and outputs $Y_i = (x_{(i-1)m+1}+ z_{(i-1)m+1}, x_{(i-1)m+2}+ z_{(i-1)m+2}, \ldots, x_{im} + z_{im})$.
Applying Lemma \ref{lem:SF2009} to channels of this alphabet yields
a sequence of schemes $\mathcal{F}_n(\{0,1\}^m)$
with 
\beq
\epsilon_{n,m} \underset{n\rightarrow\infty}{\rightarrow} 0 \mbox{.}
\label{eq:epsilonNMGoesToZero}
\eeq 
Observe that $\mathcal{F}_n(\{0,1\}^m)$ may be seen as a finite-extent scheme
both for the supersymbol alphabet $\{0,1\}^m$ additive noise channel as well as for the (fundamental)
binary alphabet $\{0,1\}$ additive noise channel.

By \eqref{eq:epsilonNMGoesToZero} we may choose $N(m)$ so that 
$\epsilon_{N(m),m} \underset{m\rightarrow\infty}{\rightarrow} 0$.  Denoting
$\mathcal{F}^m = \mathcal{F}_{N(m)}(\{0,1\}^m)$, this proves the lemma.
\end{proof}

The sequence of finite-extent schemes $\{\mathcal{F}^m\}_{m=1}^\infty$
form the basis of the universal achievability construction.
\begin{definition}
The universal achievability scheme of order $m$ is the repetition scheme
$\overline{\mathcal{F}^m}$ formed from the $N(m)$-extent scheme $\mathcal{F}^m$.
\end{definition}

We end with an important lemma regarding repetition schemes.

\begin{lemma}
\label{lem:markov-relation}
Let $\mathcal{F}$ be an $N$-extent FE scheme and let
$\overline{\mathcal{F}}$ be the corresponding repetition scheme.
Define $E_i$ be the error indicator for the $i$th block and
define $T_i = \mathds{1}_{L_i \geq a_i}$ so as to indicate if in the $i$th block
the number of bits transmitted exceeds a fixed threshold $a_i$.
Then the Markov relations 
\beq 
E_i - M_{(i)}^N - E^{i-1}
\label{eq:error-markov}
\eeq
and
\beq
T_i - M_{(i)}^N - T^{i-1}
\label{eq:rate-markov}
\eeq
both hold, where $M_{(i)}^N \equiv M_{\sum_{j=1}^{i-1} L_i + 1}^{\sum_{j=1}^{i-1} L_i + N}$
denotes the $N$ source samples used in the $i$th block by the repetition scheme.
\end{lemma}
\begin{proof}
See Appendix \ref{app:markov-relation-proof}.
\end{proof}

\subsection{Proving Achievability (Theorem \ref{thm:achievability})}
Two additional lemmas, regarding the limiting behavior of
random binary sequences, are required in order to prove that 
$\overline{\mathcal{F}^m}$ achieves the performance promised
in Theorem \ref{thm:achievability}.

\begin{lemma}
\label{lem:lawOfLargerNumbers}
Suppose $\{X_i\}$ is a sequence of iid Bernoulli$(p)$ random
variables, and suppose $\{\alpha_i\}$ is a bounded sequence
of real numbers.  Then with probability one, 
\[
\limsup_{n\rightarrow\infty}\frac{1}{n} \sum_{i=1}^n X_i \alpha_i = 
p \limsup_{n\rightarrow\infty}\frac{1}{n} \sum_{i=1}^n \alpha_i
\]
and
\[
\liminf_{n\rightarrow\infty}\frac{1}{n} \sum_{i=1}^n X_i \alpha_i = 
p \liminf_{n\rightarrow\infty}\frac{1}{n} \sum_{i=1}^n \alpha_i \mbox{.}
\]
\end{lemma}

\begin{proof}
Let $Y_i = \alpha_i X_i$.  Since $\alpha_i$ is bounded, there exists
constant $A$ such that $|\alpha_i| < A$.  We then have that
$\E{Y_i^2} \leq A^2$ is bounded and
$\sum_{k=1}^\infty \frac{1}{k^2} \var{Y_i} \leq A^2 \sum_{k=1}^\infty \frac{1}{k^2} < \infty$.
These two statements qualify the use of Kolmogorov's strong law,
which states that 
$\frac{1}{n}\sum_{i=1}^nY_i - \frac{1}{n}\sum_{i=1}^n \alpha_i p \underset{n\rightarrow\infty}{\rightarrow} 0$
with probability one.  This proves the lemma.
\end{proof}

\begin{lemma}
\label{lem:limitBounds}
Let $\{\alpha_i\}$ be a bounded real-valued sequence, and let
$\{X_i\}$ be a random binary process.  
If $\P{X_i=1 | X^{i-1}=x^{i-1}} \leq p$ for any $x^{i-1} \in \{0,1\}^{i-1}$,
then with probability one
\beqa
\label{eq:limitBoundsUpperSup}
\limsup_{n\rightarrow\infty} \frac{1}{n} \sum_{i=1}^n \alpha_i X_i \leq p\overline{\alpha} \\
\label{eq:limitBoundsUpperInf}
\liminf_{n\rightarrow\infty} \frac{1}{n} \sum_{i=1}^n \alpha_i X_i \leq p\underline{\alpha}
\eeqa
where $\overline{\alpha} = \limsup_{n\rightarrow\infty} \sum_{i=1}^n \alpha_i$ and
$\underline{\alpha} = \liminf_{n\rightarrow\infty} \sum_{i=1}^n \alpha_i$.
Similarly, if $\P{X_i=1|X^{i-1}=x^{i-1}} \geq p$ for any $x^{i-1}\in \{0,1\}^{i-1}$,
\beqa
\limsup_{n\rightarrow\infty} \frac{1}{n} \sum_{i=1}^n \alpha_i X_i & \geq & p\overline{\alpha} 
\label{eq:limitBoundsLowerSup}\\
\liminf_{n\rightarrow\infty} \frac{1}{n} \sum_{i=1}^n \alpha_i X_i & \geq & p\underline{\alpha} 
\label{eq:limitBoundsLowerInf}
\eeqa
\end{lemma}

\begin{proof}
Let $\{Y_i\}$ be a Bernoulli-$p$ iid process.  We will construct correlated binary
processes $\{\widetilde{X}_i,\widetilde{Y}_i\}$ whose marginal distributions are
identical to those of $\{X_i\}$ and $\{Y_i\}$, and use these to prove the lemma.
\begin{itemize}
\item Let $\{U_i\}$ be a sequence of independent uniform $[0,1]$ random variables.
\item For each $i$ let 
\[
\widetilde{X}_i = \left\{ \begin{array}{ll}
														1 & \text{if } U_i < \P{X_i = 1 | X^{i-1}=\widetilde{X}^{i-1}} \\
														0 & \text{otherwise}
													\end{array}
									\right.
\]
and
\[
\widetilde{Y}_i = \left\{ \begin{array}{ll}
														1 & \text{if } U_i \leq p \\
														0 & \text{otherwise}
													\end{array}
									\right.
\]
\item By Lemma \ref{lem:lawOfLargerNumbers}, 
$\limsup_{n\rightarrow\infty} \frac{1}{n} \sum_{i=1}^n \alpha_i \widetilde{Y}_i = p \overline{\alpha}$
and $\liminf_{n\rightarrow\infty} \frac{1}{n} \sum_{i=1}^n \alpha_i \widetilde{Y}_i = p \underline{\alpha}$.

\item Consider the case where $\P{X_i=1|X^{i-1}=x^{i-1}} \leq p$.
Then since $\widetilde{X}_i \leq \widetilde{Y}_i$  
we have that
\[ \limsup_{n\rightarrow\infty} \frac{1}{n} \sum_{i=1}^n \alpha_i \widetilde{X}_i \leq p \overline{\alpha}\]
and 
\[\liminf_{n\rightarrow\infty} \frac{1}{n} \sum_{i=1}^n \alpha_i \widetilde{X}_i = p \underline{\alpha}\]
with probability one.
Since $\{X_i\}$ has the same marginal distribution as $\{\widetilde{X}_i\}$, this proves
\eqref{eq:limitBoundsUpperSup} and \eqref{eq:limitBoundsUpperInf}.

\item Similarly, in the case where $\P{X_i=1|X^{i-1}=x^{i-1}} \geq p$, 
the relation $\widetilde{X}_i \geq \widetilde{Y}_i$ always holds and we have
\eqref{eq:limitBoundsLowerSup} and \eqref{eq:limitBoundsLowerInf}.
\end{itemize}
\end{proof}

\subsubsection*{Proof of Theorem \ref{thm:achievability}}
Suppose $\overline{\mathcal{F}^m}$ is applied to a source $M^\infty$
and noise sequence $z^\infty$.  Let $\{L_i\}$ be the number of source bits
decoded in each block.  Consider the $i$th block --- i.e.
$\mathcal{F}\left( (M_{\sum_{j=1}^{i-1}L_i + 1}^{\sum_{j=1}^{i-1}L_i+N(m)},0^\infty),
z_{(i-1)N(m)+1}^{iN(m)}\right)$.
Let $M_{(i)}^{N(m)} = M_{\sum_{j=1}^{i-1}L_i + 1}^{\sum_{j=1}^{i-1}L_i+N(m)}$ indicate the source
 bits used by the encoder, 
 let $\widehat{M}_{(i)}^{L_i} = \widehat{M}_{\sum_{j=1}^{i-1}L_i + 1}^{\sum_{j=1}^{i}L_i}$
indicate the estimate produced, 
let $R_i = \frac{L_i}{N(m)}$ indicate the rate of the block,
and let $z_{(i)}^{N(m)} = z_{(i-1)N(m)+1}^{iN(m)}$ denote the noise
of the block. Finally, let $u_{(i)}^{N(m)}=u_{(i-1)N(m)+1}^{iN(m)}$ denote the feedback
during the block. \\
\\

\noindent {\bf Rate Guarantees}

Recalling that $\epsilon_m$ denotes the rate-error guarantee in Corollary \ref{cor:FEConstruction},
define the ``rate indicator'' $T_i$ as the indication of whether the
rate in the $i$th block exceeds the threshold $1-\frac{1}{m}\hat{H}^m((z^{N(m)}_{(i)})) - \epsilon_m$,
i.e. 
\beq
T_i = \mathds{1}_{R_i \geq 1-\frac{1}{m}\hat{H}^m((z^{N(m)}_{(i)})) - \epsilon_m} \mbox{.}
\label{eq:rate-indicator}
\eeq
One may demonstrate that for any $i$ and any $t^{i-1} \in \{0,1\}^{i-1}$, 
$\P{T_i | T^{i-1}=t^{i-1}} \geq 1-\epsilon_m$:
\beqa
\P{T_i=1 | T^{i-1} = t^{i-1}} & = & 
\sum_{m^{N(m)} \in \{0,1\}^{N(m)}}
\left[ \P{T_i=1 | T^{i-1} = t^{i-1}, M^{N(m)}_{(i)} = m^{N(m)}} \cdot \right.\nonumber \\
& & \left. \P{M^{N(m)}_{(i)} = m^{N(m)} | T^{i-1} = t^{i-1}} \right] \nonumber \\
& \eqlabel{a} & \sum_{m^{N(m)} \in \{0,1\}^{N(m)}}
\left[ \P{T_i=1 | M^{N(m)}_{(i)} = m^n}  \P{M^{N(m)}_{(i)} = m^{N(m)} | T^{i-1} = t^{i-1}} \right] \nonumber \\
& \geqlabel{b} & 1-\epsilon_m \label{eq:error-probability-bound} \mbox{,}
\eeqa
where step (a) follows from the Markov relation (Lemma \ref{lem:markov-relation})
and step (b) is due to the rate guarantee in Corollary \ref{cor:FEConstruction}.

Applying Lemma \ref{lem:limitBounds} to $\{T_i\}$, \eqref{eq:limitBoundsLowerSup} bounds
the limit supremum rate of the scheme.
\beqa
\limsup_{k\rightarrow \infty} \frac{1}{k} \sum_{i=1}^k R_i & \geqlabel{a} & 
\limsup_{k\rightarrow \infty} \frac{1}{k} \sum_{i=1}^k R_i T_i \nonumber \\
& \geqlabel{b} & \limsup_{k\rightarrow\infty} 
\frac{1}{k} \sum_{i=1}^k T_i \left(1-\hat{H}^{m}(z_{(i)}^{N(m)}) - \epsilon_m \right) \nonumber \\
& \geqlabel{c} & (1-\epsilon_m) \limsup_{k\rightarrow\infty} 
\frac{1}{k} \sum_{i=1}^k  \left( 1 - \hat{H}^m(z_{(i)}^{N(m)}) - \epsilon_m \right) \text{ w.p.}1 \nonumber \\
& \geqlabel{d} & (1-\epsilon_m) \limsup_{k\rightarrow\infty} 
\left( 1 - \hat{H}^m(z^{kN(m)}) - \epsilon_m \right)  \text{ w.p.}1 \nonumber \\
& \geqlabel{e} & (1-\epsilon_m)(1-\underline{\delta}_m(z^\infty) - \underline{\rho}(z^\infty))  \text{ w.p.}1 
\mbox{,}
\label{eq:achieveRateGuarantee}
\eeqa
where step (a) is due to $T_i \leq 1$, step (b) follows from the definition of $T_i$,
step (c) is an application of \eqref{eq:limitBoundsLowerSup} from Lemma \ref{lem:limitBounds},
step (d) comes from the concavity $\cap$ of the entropy function,
and step (e) involves the definition 
\[ \underline{\delta}_m(z^\infty) = \epsilon_m + \liminf_{k\rightarrow\infty} \frac{1}{m}\hat{H}^m(z^{kN(m)})- \underline{\rho}(z^\infty)
\mbox{.}
\]
Since 
$\liminf_{k\rightarrow\infty} \hat{H}^m(z^{kN(m)}) =
\liminf_{k\rightarrow\infty} \hat{H}^m(z^{k})$
and by Lemma \ref{lem:compressibility}
$\underline{\rho}(z^\infty) = \lim_{m\rightarrow\infty}
\liminf_{k\rightarrow\infty} \frac{1}{m}\hat{H}^m(z^{k})$,
we have that $\underline{\delta}_m(z^\infty)$ vanishes with increasing $m$.

A similar line of logic can demonstrate the limit infimum rate bound:
\beqa
\liminf_{k\rightarrow \infty} \frac{1}{k} \sum_{i=1}^k R_i & \geqlabel{a} & 
\liminf_{k\rightarrow \infty} \frac{1}{k} \sum_{i=1}^k R_i T_i 
\nonumber 
\\
& \geq & \liminf_{k\rightarrow\infty} 
\frac{1}{k} \sum_{i=1}^k T_i \left(1-\hat{H}^{m}(z_{(i)}^{N(m)}) - \epsilon_m \right) 
\nonumber 
\\
& \geq & (1-\epsilon_m) \liminf_{k\rightarrow\infty} 
\frac{1}{k} \sum_{i=1}^k  \left( 1 - \hat{H}^m(z_{(i)}^{N(m)}) - \epsilon_m \right)  \text{ w.p.}1 
\nonumber 
\\
& \geq & (1-\epsilon_m) \liminf_{k\rightarrow\infty} 
\left( 1 - \hat{H}^m(z^{kN(m)}) - \epsilon_m \right)   \text{ w.p.}1 
\nonumber 
\\
& \geq & (1-\epsilon_m)
\left(1- \overline{\rho}(z^\infty)- \overline{\delta}_m(z^\infty) \right)  \text{ w.p.}1 
\mbox{,}
\label{eq:overachieveRateGuarantee}
\eeqa

where in the final step we define
\[ \overline{\delta}_m(z^\infty) = \epsilon_m + \limsup_{k\rightarrow\infty} \hat{H}^m(z^{kN(m)})- \overline{\rho}(z^\infty)
\mbox{.}
\]
As in the infimum case, since 
$\limsup_{k\rightarrow\infty} \hat{H}^m(z^{kN(m)}) =
\limsup_{k\rightarrow\infty} \hat{H}^m(z^{k})$
and, by Lemma \ref{lem:compressibility}, \\
$\overline{\rho}(z^\infty) = \lim_{m\rightarrow\infty}
\limsup_{k\rightarrow\infty} \frac{1}{m}\hat{H}^m(z^{k})$,
we have that $\overline{\delta}_m(z^\infty)$ vanishes as $m\rightarrow\infty$.

{\bf Error Guarantees}
Let $E_i$ indicate the presence of an error in the $i$th block,
i.e. $E_i = \mathds{1}_{\widehat{M}_{(i)}^{L_i}= M_{(i)}^{L_i}}$.
The limit-supremum bit-error rate may be written in terms of $\{E_i\}$
and $\{L_i\}$ as
\beqa
\limsup_{n\rightarrow\infty} \frac{\sum_{i=1}^n E_i L_i}{\sum_{i=1}^n L_i}
& = & 
\limsup_{n\rightarrow\infty} \frac{\frac{1}{n}\sum_{i=1}^n E_i L_i}{\frac{1}{n}\sum_{i=1}^n L_i} 
\nonumber \\
& \leq & \frac{\limsup_{n\rightarrow\infty} \frac{1}{n}\sum_{i=1}^n E_i L_i}
{\liminf_{n\rightarrow\infty} \frac{1}{n}\sum_{i=1}^n L_i} 
\nonumber \\
& \leqlabel{a} & \frac{\limsup_{n\rightarrow\infty} \frac{1}{n}\sum_{i=1}^n E_i N(m)}
{\liminf_{n\rightarrow\infty} \frac{1}{n}\sum_{i=1}^n L_i} 
\nonumber \\
& \leqlabel{b} & \frac{N(m) \limsup_{n\rightarrow\infty} \frac{1}{n}\sum_{i=1}^n E_i }
{N(m)\left(1-\overline{\rho}(z^\infty) - \overline{\delta}_m(z^\infty) \right)} \text{ w.p.1} \mbox{,}
\label{eq:errorBoundPartway}
\eeqa
where (a) holds because $L_i \leq N(m)$ for an $N(m)$-horizon
repetition scheme. (b) follows with probability one from \eqref{eq:overachieveRateGuarantee}.

As was done with $\{T_i\}$, one may demonstrate
that for any $i$ and $e^{i-1} \in \{0,1\}^{i-1}$, the bound
$\P{E_i=1 | E^{i-1} = e^{i-1}} \leq \epsilon_m$ holds:
\beqa
\P{E_i=1 | E^{i-1} = e^{i-1}} & = & 
\sum_{m^{N(m)} \in \{0,1\}^{N(m)}}
\left[ \P{E_i=1 | E^{i-1} = e^{i-1}, M^{N(m)}_{(i)} = m^{N(m)}} \cdot \right.\nonumber \\
& & \left. \P{M^{N(m)}_{(i)} = m^{N(m)} | E^{i-1} = e^{i-1}} \right] \nonumber \\
& \eqlabel{a} & \sum_{m^{N(m)} \in \{0,1\}^{N(m)}}
\left[ \P{E_i=1 | M^{N(m)}_{(i)} = m^n}  \P{M^{N(m)}_{(i)} = m^{N(m)} | E^{i-1} = e^{i-1}} \right] \nonumber \\
& \leqlabel{b} & \epsilon_m \mbox{,}
\eeqa
where step (a) follows from the Markov relation (Lemma \ref{lem:markov-relation})
and step (b) is due to the error bound in Corollary \ref{cor:FEConstruction}.
This allows for the application of Lemma \ref{lem:limitBounds} to $\{E_i\}$ with constant weights
$\alpha_i = 1$, establishing that
\beq
\limsup_{n\rightarrow\infty} \frac{1}{n} \sum_{i=1}^n E_i \leq \epsilon_m \mbox{.}
\label{eq:block-error-rate}
\eeq
This in turn may be inserted into \eqref{eq:errorBoundPartway}, proving
\[
\limsup_{n\rightarrow\infty} \frac{ \sum_{i=1}^n E_i L_i}{\sum_{i=1}^n L_i} \leq
\frac{\epsilon_m}{1-\overline{\rho}(z^\infty) - \overline{\delta}_m(z^\infty)} \text{ w.p.}1\mbox{.}
\]
Therefore scheme $\overline{\mathcal{F}^m}$ worst-case (and best-case) achieves bit-error rate $\epsilon_m$.

\section{Predictability and a Simpler Sub-Optimal Scheme}
\label{sec:predictability}
The finite-state \emph{predictability} was introduced by Feder et al. in \cite{FederMG1992}
as an analog of compressibility in the context of universal prediction,
just as porosity is an analog in the context of modulo-additive channels.
We explore the relationship between porosity and predictability, but we do so
with fairly pragmatic motivations.

\subsection{Practicality of $\{\overline{\F^m}\}$}
While the achievability schemes $\{\overline{\mathcal{F}^m}\}$ manage to asymptotically achieve
porosity for any sequence, they are not particularly simple to implement.  The complexity of $\overline{\mathcal{F}^m}$
is hidden within the Shayevitz-Feder empirical-capacity-achieving scheme at its core (Corollary \ref{cor:FEConstruction}).
At each time instant, the Shayevitz-Feder decoder is required to compute the posterior of the message given
all the channel outputs in the block so far.
This computation is linear in the alphabet size, but because $\overline{\mathcal{F}^m}$ applies
Shayevitz-Feder to binary $m$-tuples (Corollary \ref{cor:FEConstruction}), it is exponential in $m$.

Although it grows in complexity quite rapidly, this Horstein-based approach of Shayevitz and Feder is actually quite
efficient for small alphabets, e.g. binary.  The only reason it is applied to $m$-tuples in our construction
is to account for memory and correlation within the noise sequence.  Alternatively stated, we seek to 
achieve the $m$th-order empirical capacity for arbitrarily large $m$.  The simpler repetition schemes suggested
in this section take a layered approach, wherein memory and correlation is first ``removed'' from the noise sequence,
after which the binary-alphabet Shayevitz-Feder scheme is used to communicate with the decoder.


\subsection{Construction of layered scheme}
In defining the more practical repetition
schemes $\overline{\G_n}$, we first describe the finite-extent schemes $\{\G_n\}$ at their heart. 
As shown in Fig. \ref{fig:predictive-schemes}, in its inner layer $\G_n$ consists of a predictor that forms an estimate 
$\hat{z}_i$ for the noise $z_i$ at each time $i \in \{1,\ldots,n\}$.
By subtracting this prediction from the encoder's output, a ``surrogate'' channel is created with
``effective'' noise $\widetilde{z}^n = (z_i - \hat{z_i})_{i=1}^n$.  In other words,
the noise $z^n$ is replaced by a sequence of error indications for the predictor (Fig. \ref{fig:predictive-effective}).
The first-order finite-extent scheme $\mathcal{F}_n(\{0,1\})$ (defined in Lemma \ref{lem:SF2009}),
which we will refer to simply as $\mathcal{F}_n$,
is then applied to this surrogate channel. Roughly speaking, the prediction step exploits
the memory and correlation in $z^n$ in order to reduce its first-order empirical entropy,
which then serves to boost the performance of $\mathcal{F}_n$.

\begin{figure}
 \centering
  \psfrag{c}[cc][cc]{\small Unknown Channel}
  \psfrag{m}[cc][cc]{$M^\infty$}
  \psfrag{e}[cc][cc]{E}
  \psfrag{d}[cc][cc]{D}
  \psfrag{m2}[cc][cc]{$\widehat{M}^L$}
  \psfrag{x}[cc][cc]{$x_i$}
  \psfrag{y}[cc][cc]{$y_i$}
  \psfrag{z}[cc][cc]{$z^{-1}$}
  \psfrag{zhat}[cc][cc]{$-\hat{z}_i$}
  \psfrag{y-1}[cc][cc]{Predictor}
  \psfrag{2}[cc][cc]{$2$}
  \psfrag{noise}[cc][cc]{$z_i$}
  \psfrag{+}[cc][cc]{$+$}
  \psfrag{t}[cc][cc]{$\theta$}
  \includegraphics[width=3.8in]{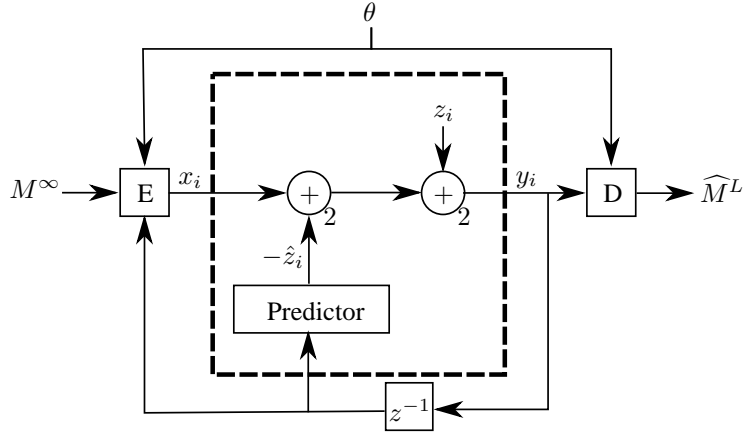}
  \caption{Block diagram for the scheme $\G_n$.}
 \label{fig:predictive-schemes}
\end{figure}

\begin{figure}
 \centering
  \psfrag{c}[cc][cc]{\small Unknown Channel}
  \psfrag{m}[cc][cc]{$M^\infty$}
  \psfrag{e}[cc][cc]{E}
  \psfrag{d}[cc][cc]{D}
  \psfrag{m2}[cc][cc]{$\widehat{M}^L$}
  \psfrag{x}[cc][cc]{$x_i$}
  \psfrag{y}[cc][cc]{$y_i$}
  \psfrag{z}[cc][cc]{$z^{-1}$}
  \psfrag{y-1}[cc][cc]{$y_{i-1}$}
  \psfrag{2}[cc][cc]{$2$}
  \psfrag{noise}[cc][cc]{$\tilde{z}_i$}
  \psfrag{+}[cc][cc]{$+$}
  \psfrag{t}[cc][cc]{$\theta$}
  \includegraphics[width=3.8in]{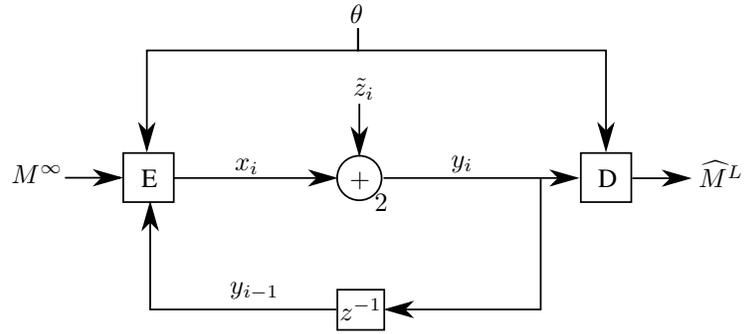}
  \caption{An alternative block diagram for scheme $\G_n$, where
  the prediction loop is represented as a surrogate noise sequence $\tilde{z}^\infty$.}
 \label{fig:predictive-effective}
\end{figure}

To aid in formally constructing $\{\G_n\}$, define the following:
\begin{itemize}
\item $\{e_i^{\F}\}$ refers to the encoding functions for $\mathcal{F}_n$.
\item $d_L^\F$ and $d_M^\F$ refer to the decoding-length and decoding-message functions of $\F_n$.
\item $\{p_{i}(u_i | u^{i-1},y^i)\}$ is the set of feedback conditional distributions for $\mathcal{F}_n$, and $\theta_{\mathcal{F}}$
is the common randomness.
\item $\hat{z}_i = f(z^{i-1})$ is the estimation function implemented by the prediction scheme (which has yet to be described).
\end{itemize}
The prediction-based FE scheme $\G_n$ can then be defined in terms of the six components listed in the definition of FE schemes:
\begin{enumerate}
\item The extent is $n$.
\item The feedback channel is a simple delay-one noiseless feedback: $u_i = y_{i-1}$.
\item The common randomness is given by $\theta_i = (\theta_{\mathcal{F}}, U^{(1)}, U^{(2)}, \ldots, U^{(n)})$.
$U^{(i)}$ is a $\mathcal{X}^{2^{2i-1}}$-valued vector indexed by $(y^{i},u^{i-1})$.  The $(y^i,u^{i-1})$-th component 
$U^{(i)}_{y^i,u^{i-1}}$ is distributed according to $p_{i}(\cdot | y^i,u^{i-1})$.  Note that
$U^{(i)}$ is introduced only to simulate the feedback channel of scheme $\F_n$ at both encoder and decoder.
This is similar to the technique used in the proof of Lemma \ref{lem:FSAFSisFSS}.
\item To define the $i$th encoding function, first let $u_i^{\F} = U^{(i)}_{y^{i},u^{i-1,\F}}$ be
the simulated feedback channel.  Note that both encoder and decoder may compute $u_i^{\F}$
at the end of the $i$th time step.  The encoding functions are then given by
\[ e_i(M^\infty,\theta,u^{i-1}) = 
e_i^\F(M^\infty, \theta_{\F}, u^{i-1,\F}) - f(y^{i-1} - x^{i-1})
\mbox{.}
\]
\item The decoding length function is unchanged: $L = d_L^\F(y^n,\theta,u^{n-1,\F})$.
\item The decoding function is also unchanged: $\widehat{M}^L = d_M^\F(y^n,\theta,u^{n-1,\F})$.
\end{enumerate}

The prediction scheme at the heart of $\G_n$ is the \emph{incremental parsing} (IP) algorithm
introduced by Feder et al. \cite{FederMG1992}.  This is an elegant and simple algorithm that is based on the same parsing procedure
as Lempel and Ziv's compression scheme.  Rather than describe its operation in detail, we point the
reader towards the exposition in Sec. V of \cite{FederMG1992}.

\subsection{Analysis of performance}
The schemes $\{\overline{\G_n}\}$ have been constructed to simplify the encoding and decoding process.
Here, this notion is quantified.

As a layered scheme, $\overline{\G_n}$ consists of two machines running in parallel: the IP
predictor of Feder et al. \cite{FederMG1992} and the actual communication scheme $\F_n$
of Shayevitz and Feder \cite{ShayevitzF2009}.  The operations-per-time-step required by
$\F_n$ do not scale appreciably with $n$, so the bottleneck is the prediction operation.  
Observe that the complexity bottleneck for $\F^m$ also arises from accounting for the noise
sequence's memory and correlation.  As previously mentioned, accounting for memory with $\F^m$ requires an exponential number of
operations-per-time-step. The IP predictor, on the other hand, requires only a linear number of operations 
in order to produce an estimate.  Specifically, at each time step the Lempel-Ziv parsing tree must be extended.

The rates achieved by schemes $\G_n$ however do not quite reach porosity.
To illustrate this, we start by repeating the definition of predictability as given in \cite{FederMG1992}.
\begin{definition}
The finite-state predictability of a sequence $x^\infty$ is the minimum limit-supremum
fraction of errors that a finite-state predictor can attain when operating on $x^\infty$.
Just as with compressibility, one may define a limit infimum version of this quantity.  
We term the former the worst-case predictability $\overline{\pi}$ and the latter the
best-case predictability $\underline{\pi}$.
\end{definition}

In Theorem 4 of \cite{FederMG1992} it is shown that the IP predictor achieves
the worst-case predictability $\overline{\pi}(x^\infty)$ of any sequence $x^\infty$.
Though it is not stated in the theorem, the proof that is given also demonstrates that
the IP predictor achieves the best-case predictability $\underline{\pi}(x^\infty)$.
Therefore the limit supremum (or infimum) first-order empirical entropy of the 
surrogate noise sequence approaches $h_b(\overline{\pi}(x^\infty))$ (or $h_b(\underline{\pi}(x^\infty))$). 
By applying the FS schemes $\overline{\F_n}$ to this noise sequence, the 
performance approaches rate $1-h_b(\overline{\pi}(x^\infty))$ (or $1-h_b(\underline{\pi}(x^\infty))$)
with vanishing error.

In Sec. VI of \cite{FederMG1992}, the worst-case predictability $\overline{\pi}(x^\infty)$ is bounded in terms of the 
 compressibility:
\[
h_b^{-1}(\overline{\rho}(x^\infty)) \leq \overline{\pi}(x^\infty) \leq \frac{1}{2} \overline{\rho}(x^\infty) \mbox{.}
\]
 An identical set of bounds exist between the best-case predictability and best-case compressibility.
 When a noise sequence satisfies the lower bound with equality, one may observe that the asymptotic performance of 
 $\{\overline{\G_n}\}$ matches that of $\{\overline{\F^m}\}$.  However, this is usually not the case, and one must settle
 for the guarantee of worst-case rate $1-h_b(\overline{\rho}(z^\infty)/2)$ and
 best-case rate $1-h_b(\underline{\rho}(z^\infty)/2)$.  Each falls strictly below worst- and best-case porosity
 unless the noise sequence is either completely redundant or incompressible.

\section{Summary}
\label{sec:Conclusion}
In this work, the best-case/worst-case porosity $\overline{\sigma}( \cdot )/\underline{\sigma}(\cdot) $
 of a binary noise sequence is defined as one minus the best-case/worst-case compressibility 
 $1-\underline{\rho}(\cdot ) / 1-\overline{\rho}(\cdot )$.  Porosity may be seen as an individual sequence
 property, analogous to compressibility or predictability, that identifies the ease of communication through a
 modulo-additive noise sequence.  Two results regarding porosity are at the core of this work.  First,
 porosity is identified as the maximum achievable rate within the class of finite-state communication schemes.
 Second, it is shown that porosity may be universally achieved within this class.  Together, these results
 parallel those of Lempel and Ziv in the source coding context \cite{ZivL1978}.
 Furthering this analogy, the achievability schemes given here complement those of
 Lomnitz and Feder \cite{LomnitzF2011} in similar manner as the infinite-state and finite-state schemes of
 \cite{ZivL1978}.
 
 In addition to the above, a more practical universal communication architecture is introduced, built upon 
 prediction.  Rather than communicate using blocks of channel uses --- which contributes to an exponentially
 growing complexity --- a layered approach is taken.  A prediction algorithm first ``removes'' the memory from the noise,
 and then a simple first-order communication scheme is employed.  While the resulting algorithm is suboptimal, 
 it reduces complexity considerably, and also draws an operational connection between predictability and porosity.

\appendices
\section{Proof of Lemma \ref{lem:compressibility}}
\label{app:compressibility-proof}
We show that the distinction between block-by-block and sliding-window
empirical entropy computations vanishes in the limit of large blocks and long blocklengths.
First, a few definitions that simplify notation:
\begin{definition}
A $k$-block code $C$ maps $k$-tuples from an alphabet $\mathcal{X}^k$
into binary strings of arbitrary but finite length.
\end{definition}
\begin{definition}
\label{def:extension-code}
The $(\theta, \tilde{k})$-\emph{extension code} for a $k$-block code
$C$ is a $\kt$-block code $\Ct_\theta$ whose encoding of a block $X^{\kt}$
is given by
\[
\Ct_{\theta}(X^\kt) = 
		\left(X_1^\theta, 
					C(X_{\theta+1}^{\theta+k}),
					\ldots,
					C(X_{\lfloor \frac{\kt-\theta}{k}\rfloor k - k - 1 + \theta}^{\lfloor \frac{\kt-\theta}{k}\rfloor k + \theta}),
					X_{\lfloor \frac{\kt-\theta}{k}\rfloor k + \theta+1}^{\kt}
		\right)
\mbox{.}
\]
The $\kt$-\emph{extension code} is a $\kt$-block code whose
encoding is given by
\[
\Ct(X^\kt) = 
\underset{\Ct_{\theta}(X^\kt), \theta \in \{0,\ldots,k-1\}} {\operatorname{argmin}}
\ell\left( \Ct_{\theta} (X^\kt) \right) 
\mbox{,}
\]
where $\ell(\cdot)$ returns the length of a binary string.
\end{definition}
In both the $\kt$-extension and the $(\theta,\kt)$-extension,
the initial segment of $X_1^\theta$ is referred to as the $\emph{uncoded prefix}$,
while the punctuating segment $X_{\lfloor \frac{\kt-\theta}{k}\rfloor k + \theta+1}^{\kt}$
is the \emph{uncoded suffix}.
The encoded segments in the middle are called the \emph{encoded subblocks}.

We start by demonstrating that the block-by-block empirical entropy limits exist.
\begin{lemma}
\label{lem:block-by-block-limit-exists}
Let $x^\infty$ be a finite-alphabet sequence.  Then the limits
\[
\lim_{k\rightarrow\infty}\limsup_{n\rightarrow\infty} \frac{1}{k} \hat{H}^k(x^n)
\]
and
\[
\lim_{k\rightarrow\infty}\liminf_{n\rightarrow\infty} \frac{1}{k} \hat{H}^k(x^n)
\]
both exist.
\end{lemma}

\begin{proof}
Let $\kt > k$, and let $C_{n,k}$ denote the k-block Huffman code
for the block-by-block empirical distribution $\hat{p}(X^k)[x^n]$.
Observe that by optimality of the Huffman code, 
\beq
\hat{H}^k(x^n) \leq \E{\ell(C_{n,k}(X^k))}_{\hat{p}(X^k)[x^n]} = \E{\ell(C_{n,k}(x_{kN+1}^{kN+k}))} \leq \hat{H}^k(x^n)+1
\label{eq:HuffmanGuarantee}
\eeq
where $N$ is distributed uniformly over the set $\{0,\ldots,\lfloor \frac{n}{k} \rfloor - 1\}$.
Let $\Ct_{n,\kt,\theta}$ and $\Ct_{n,\kt}$ be the $(\theta,\kt)$- and $\kt$-extensions of $C_{n,k}$.

By expressing the expected length of $\Ct_{n,\kt}$ in
terms of the expected length of $C_{n,k}$, we can show that both
limits in the lemma statement exist.

Allowing $M$ to be uniformly distributed over the set
$\{0,\ldots,\lfloor \frac{n}{\kt} \rfloor - 1\}$, we have that
\beqan
\widehat{H}^\kt(x^n) & \leqlabel{a} & \E{ \ell \left( \Ct_{n,\kt}(x_{\kt M+1}^{\kt M + \kt}) \right)} \\
& \eqlabel{b} & \E{ \min_{\theta \in \{0,\ldots,k-1\}} \ell \left( \Ct_{n,\kt,\theta}(x_{\kt M+1}^{\kt M+\kt}) \right)} \\
& \leqlabel{c} & \E{ \ell \left( \Ct_{n,\kt,-M\kt\bmod k}(x_{\kt M+1}^{\kt M+\kt}) \right)} \\
& \leqlabel{d} & \E{2k + \sum_{m=0}^{\lfloor\frac{n}{k}\rfloor- 1}
													\ell\left( C_{n,k}(x_{km+1}^{km+k})\right) 
													{\bf 1}_{km+1\geq \kt M+1} 
													{\bf 1}_{km+k\leq \kt M+\kt}
													} \\
& = & 2k + \sum_{m=0}^{\lfloor\frac{n}{k}\rfloor- 1}
								\P{km +1\geq \kt M+1, km + k \leq \kt M + \kt}
								\ell\left(C_{n,k}(x_{km+1}^{km+k})\right) \\
& \leqlabel{e} & 2k + \sum_{m=0}^{\lfloor\frac{n}{k}\rfloor- 1}
											 \frac{\frac{\kt}{k}}{\frac{n}{k}-1}
								\ell\left(C_{n,k}(x_{km+1}^{km+ k})\right) \\
& \eqlabel{f} & 2k + \left\lfloor \frac{n}{k} \right\rfloor \frac{\kt}{n-k} \E{\ell(C_{n,k}(x_{kN+1}^{kN+k}))} \\
& \leqlabel{g} & 2k + \left\lfloor \frac{n}{k} \right\rfloor \frac{\kt}{n-k}\left(\widehat{H}^k(x^n) +1 \right)
\eeqan
Step (a) follows by Shannon's source coding converse.
(b) is from the definition of $\Ct_{n,\kt}$.
(c): Observe that by setting $\theta = -M\kt\bmod k$, every encoded subblock within
$\Ct_{n,\kt,\theta}$ is aligned so that it will be of the form $C_{n,k}(x_{km+1}^{k(m+1)})$ for some integer $m$.
Step (d) involves first upper-bounding the length of the unencoded
prefix and suffix components of $\Ct_{n,k}(x_{kM+1}^{k(M+1)})$ at 
$k$ bits each,
and then summing the lengths of each of the encoded subblocks.
Note that an encoded subblock $C_{n,k}(x_{km+1}^{km+k})$ only appears in
$\Ct_{n,\kt,-M\kt \bmod k} (x_{\kt M+ 1}^{\kt M+\kt})$
if $x_{km+1}^{km+k}$ is fully contained within the $\kt$-block being encoded
$x_{\kt M+1}^{\kt M + \kt}$.  The indicator functions ensure that only these encoded
subblocks contribute to the length summation.
(e) follows from recognizing that $M$ takes $\lfloor n/k\rfloor$ 
values uniformly, and that at most $\kt/k$ of these positions satisfy the conditions
$mk +1 \geq M\kt+1$ and $mk+k \leq M\kt + \kt$. 
(f) replaces the summation with an expectation, where the random variable
$N$ is uniformly distributed over the set $\{0,\ldots,\lfloor\frac{n}{k}\rfloor-1 \}$.
(g) invokes \eqref{eq:HuffmanGuarantee}.

Dividing both sides of this resulting inequality by $\kt$
and taking the limit supremum with respect to $n$, we have that
\[
\limsup_{n\rightarrow\infty} \frac{1}{\kt} \hat{H}^\kt(x^n) \leq 
\frac{2k}{\kt} + \frac{1}{k} + \limsup_{n\rightarrow\infty} \frac{1}{k} \hat{H}^k(x^n)
\mbox{.}
\]
Taking the limit supremum of both sides of this expression with $\kt$, we have that
\[
\limsup_{\kt \rightarrow\infty} \limsup_{n\rightarrow\infty} \frac{1}{\kt} \hat{H}^\kt(x^n) \leq 
\frac{1}{k} + \limsup_{n\rightarrow\infty} \frac{1}{k} \hat{H}^k(x^n)
\mbox{.}
\]
Finally, taking the limit infimum with respect to $k$,
\[
\limsup_{\kt \rightarrow\infty} \limsup_{n\rightarrow\infty} \frac{1}{\kt} \hat{H}^\kt(x^n) \leq 
\liminf_{k \rightarrow\infty} \limsup_{n\rightarrow\infty} \frac{1}{k} \hat{H}^k(x^n)
\mbox{.}
\]
This proves that $\lim_{k\rightarrow\infty}\limsup_{n\rightarrow\infty} \frac{1}{k} \hat{H}^k(x^n)$
exists.

Repeating this last set of arguments with the limit infimum with respect to $n$
proves that
$\lim_{k\rightarrow\infty}\liminf_{n\rightarrow\infty} \frac{1}{k} \hat{H}^k(x^n)$
exists.
\end{proof}

Next, we prove that the sliding-window compressibility can be no greater than the block-by-block
compressibility.
\begin{lemma}
\label{lem:block-by-block-upper-bound}
Let $x^\infty$ be a finite-alphabet sequence.  Then the following two statements hold:
\[
\underline{\rho}(x^\infty)
\leq \lim_{k\rightarrow\infty} \liminf_{n\rightarrow\infty} \frac{1}{k} \hat{H}^k(x^n)
\]
and
\[
\overline{\rho}(x^\infty)
\leq \lim_{k\rightarrow\infty} \limsup_{n\rightarrow\infty} \frac{1}{k} \hat{H}^k(x^n) \mbox{.}
\]
\end{lemma}

\begin{proof}
As the form of this proof is very similar to that of Lemma \ref{lem:block-by-block-limit-exists},
exposition will be limited.

Let $\kt > k$, let $C_{n,k}$ denote the k-block Huffman code
for the block-by-block empirical distribution $\hat{p}(X^k)[x^n]$,
and let $\Ct_{n,\kt,\theta}$ and $\Ct_{n,\kt}$ be the $(\theta,\kt)$- and $\kt$-extensions of $C_{n,k}$.  

Allowing $M$ to be uniformly distributed over the set $\{0,\ldots,n-\kt\}$, we have that
\beqan
\hat{H}^\kt_{\text sw} & \leq & \E{\ell\left(\Ct_{n,\kt}(x_{M+1}^{M+\kt})\right)} \\
& = & \E{ \min_{\theta \in \{0,\ldots,k-1\}}\ell\left(\Ct_{n,\kt,\theta}(x_{M+1}^{M+\kt})\right)} \\
& \leq & \E{ \ell\left(\Ct_{n,\kt,-M\bmod k}(x_{M+1}^{M+\kt})\right)} \\
& \leq & \E{2k + \sum_{m=0}^{\left\lfloor\frac{n}{k}\right\rfloor-1}
												  \ell\left(C_{n,k}(x_{km+1}^{km+k})\right)
												  {\bf 1 }_{km+1\geq M+1} 
												  {\bf 1 }_{km+k \leq M+\kt}} \\
& = & 2k + \sum_{m=0}^{\left\lfloor\frac{n}{k}\right\rfloor-1}
              \P{km+1\geq M+1,km+k \leq M+\kt}
              \ell\left(C_{n,k}(x_{km+1}^{km+k})\right) \\
&\leq & 2k + \sum_{m=0}^{\left\lfloor\frac{n}{k}\right\rfloor-1}
 							 \frac{\kt -k-1}{n-\kt-1}
 							\ell\left(C_{n,k}(x_{km+1}^{km+k})\right) \\
& \leq & 2k+ \left\lfloor\frac{n}{k}\right\rfloor \frac{\kt}{n-\kt-1}
 							\E{\ell\left(C_{n,k}(x_{kN+1}^{kN+k})\right)}, 
							N \sim \text{Unif}\left\{ 0,\ldots, \left\lfloor \frac{n}{k} \right\rfloor-1\right\} \\
& \leq & 2k + \frac{n}{k} \frac{\kt}{n-\kt-1} \left( \hat{H}^k(x^n)+1 \right)
\eeqan

Dividing both sides by $\kt$ and taking the limit supremum in $n$,
we have that
\[
\limsup_{n\rightarrow\infty}\frac{1}{\kt}\hat{H}^\kt_{\text sw}
\leq
\frac{2k}{\kt} + \frac{1}{k} \left( \hat{H}^k(x^n) +1 \right)
\mbox{.}
\]

Now taking the limit supremum in $\kt$ followed by the 
limit supremum in $k$, we have that
\beq
\limsup_{\kt\rightarrow\infty} \limsup_{n\rightarrow\infty}\frac{1}{\kt}\hat{H}^\kt_{\text sw}
\leq
\limsup_{k\rightarrow\infty} \limsup_{n\rightarrow\infty}\frac{1}{k}\hat{H}^k
\mbox{.}
\label{eq:supBound}
\eeq

We can identically show (by replacing all limits supremum with limits infimum) that 
\beq
\liminf_{\kt\rightarrow\infty} \liminf_{n\rightarrow\infty}\frac{1}{\kt}\hat{H}^\kt_{\text sw}
\leq
\liminf_{k\rightarrow\infty} \liminf_{n\rightarrow\infty}\frac{1}{k}\hat{H}^k
\mbox{.}
\label{eq:infBound}
\eeq

Equations \eqref{eq:supBound} and \eqref{eq:infBound} prove the lemma.
\end{proof}								

We now prove the opposite direction: that the sliding window compressibility can be no
smaller than the block-by-block compressibility.

\begin{lemma}
\label{lem:block-by-block-lower-bound}
Let $x^\infty$ be a finite-alphabet sequence.  Then
\[
\underline{\rho}(x^\infty)
\geq \lim_{k\rightarrow\infty} \liminf_{n\rightarrow\infty} \frac{1}{k} \hat{H}^k(x^n)
\]
and
\[
\overline{\rho}(x^\infty)
\geq \lim_{k\rightarrow\infty} \limsup_{n\rightarrow\infty} \frac{1}{k} \hat{H}^k(x^n) \mbox{.}
\]
\end{lemma}

\begin{proof}
Let $\kt>k$ be two blocklengths, and let $C_{n,k}$ be the optimal $k$-block Huffman
code for the $k$th-order sliding window distribution 
$\hat{p}^k_{\text sw}(X^k)[x^n]$.  Then there must exist a phase
$\theta^* \in \{0,\ldots,k-1\}$ such that
\beq
\E{\ell \left( C_{n,k}(x_{Nk+1+\theta^*}^{(N+1)k+\theta^*}) \right)}
\leq \hat{H}^k_{\text sw}(x^n) + 1
\label{eq:HuffmanGuarantee2}
\eeq
where $N$ is distributed uniformly across the set
$\{0,\ldots,\left\lfloor \frac{n-\theta^*}{k} \right\rfloor - 1\}$.
This follows from the within-one-bit optimality of $C_{n,k}$ over the 
sliding-window distribution, and because the sliding-window distribution
may be expressed as a (nonnegative) linear combination of the 
block-by-block distributions $\hat{p}^k(X^k)[x_\theta^n]$ 
computed according to each phase $\theta$.

Define $\widetilde{C}_{n,k,\theta}$ and $\widetilde{C}_{n,k}$
as the $(\theta,\kt)$- and $\kt$-extensions of $C_{n,k}$.

  A familiar sequence of inequalities may then be applied.
Allowing $M$ to be uniformly distributed over the set $\{0,\ldots, \lfloor n/\kt \rfloor - 1\}$,
we have:  
\beqan
\hat{H}^{\kt}(x^n) & \leqlabel{a} & \E{\ell \left( \Ct_{n,\kt} (x_{\kt M+ 1}^{\kt(M+1)}) \right)} \\
& = & \E{ \min_{\theta \in \{0,\ldots,k-1\}} \ell \left( \Ct_{n,\kt,\theta} (x_{\kt M+ 1}^{\kt(M+1)}) \right)} \\
& \leqlabel{b} & \E{ \ell \left( \Ct_{n,\kt,(\theta^*-M)\bmod k} (x_{\kt M+ 1}^{\kt(M+1)}) \right)} \\
& \leqlabel{c} & \E{ 2k + \sum_{m=0}^{\left\lfloor\frac{n-\theta^*}{k}\right\rfloor- 1}
													\ell\left( C_{n,k}(x_{km+1+\theta^*}^{k(m+1)+\theta^*})\right) 
													{\bf 1}_{M\kt+1 \leq km+\theta^*+1} 
													{\bf 1}_{M\kt+\kt \geq mk +\theta^*+k}
										} \\
& \leq & 2k + \sum_{m=0}^{\left\lfloor\frac{n-\theta^*}{k}\right\rfloor- 1}
							\P{M\kt+1 \leq km+\theta^*+1, M\kt+\kt \geq mk +\theta^*+k}
							\ell\left( C_{n,k}(x_{km+1+\theta^*}^{k(m+1+\theta^*)})\right) 
							\\
& \leqlabel{d} & 2k + \sum_{m=0}^{\left\lfloor\frac{n-\theta^*}{k}\right\rfloor- 1}
											\frac{ \frac{\kt}{k}}{\frac{n-\theta^*}{k}-1}
											\ell \left( C_{n,k}(x_{km+1+\theta^*}^{k(m+1+\theta^*)})\right) 
							\\
& \leqlabel{e} & 2k + \frac{n-\theta^*}{k}\frac{\kt}{n-\theta^*-k}
											\E{\ell \left( C_{n,k}(x_{kN+1+\theta^*}^{k(N+1)+\theta^*})\right)} 
											\\
& \leqlabel{f} & 2k + \frac{n-\theta^*}{k}\frac{\kt}{n-\theta^*-k}
											\left(\hat{H}^k_{\text sw}(x^n) + 1\right)
\eeqan
Step (a) follows from Shannon's converse and the definition of 
the $\kt$-block empirical entropy $\hat{H}^{\kt}(x^n)$.
(b) sets $\theta = (\theta^*-M)\bmod k$ so that every encoded subblock within $\Ct_{n,\kt,\theta}$
 is of the form $C_{n,k}(x_{km+1+\theta^*}^{k(m+1)+\theta^*})$ for some integer $m$.
(c) upper-bounds the length of an encoding
first by bounding the suffix and prefix at $k$ bits each, and then by
summing the encoding lengths for each contributing encoded subblock.
Observe that an encoded subblock $C_{n,k}(x_{km+1+\theta^*}^{k(m+1+\theta^*)})$ only appears in
$\Ct_{n,\kt,(\theta^*-M)\bmod k} (x_{\kt M+ 1}^{\kt(M+1)})$
if $x_{km+1+\theta^*}^{k(m+1+\theta^*)}$ is fully contained within the $\kt$-block being encoded
$x_{\kt M+1}^{\kt(M+1)}$.
(d) follows from recognizing that $M$ takes $\lfloor (n-\theta^*)/k\rfloor$ 
values uniformly, and that at most $\kt/k$ of these positions satisfy the conditions
$M\kt+1 \leq km+\theta^*+1$ and $M\kt+\kt \geq mk +\theta^*+k$. 
(e) introduces $N$ as a random variable distributed uniformly over 
$\{0,\ldots,\left\lfloor \frac{n-\theta}{k}\right\rfloor-1\}$, and 
the sum is replaced by an expectation.
Finally, (f) is a direct application of \eqref{eq:HuffmanGuarantee2}.

Dividing both sides by $\kt$ and taking the limit infimum in $n$, we have that
\[
\liminf_{n\rightarrow\infty} \frac{1}{\kt} \hat{H}^\kt(x^n)
\leq
\frac{2k}{\kt} + \frac{1}{k} \left(\hat{H}^k_{\text sw}(x^n) + 1\right)
\mbox{.}
\]
Taking the limit infimum in $\kt$, followed by the limit infimum in $k$,
we have
\beq
\liminf_{\kt\rightarrow\infty} \liminf_{n\rightarrow\infty} \frac{1}{\kt} \hat{H}^\kt(x^n)
\leq
\liminf_{k\rightarrow\infty} \liminf_{n\rightarrow\infty} \frac{1}{k} \hat{H}_{\text sw}^k(x^n)
\mbox{.}
\label{eq:InfimumBound}
\eeq

By repeating these two steps with limits supremum instead of 
limits infimum, we find that
\beq
\limsup_{\kt\rightarrow\infty} \limsup_{n\rightarrow\infty} \frac{1}{\kt} \hat{H}^\kt(x^n)
\leq
\limsup_{k\rightarrow\infty} \limsup_{n\rightarrow\infty} \frac{1}{k} \hat{H}_{\text sw}^k(x^n)
\mbox{.}
\label{eq:SupremumBound}
\eeq
\end{proof}

Lemma \ref{lem:compressibility} now follows from lemmas \ref{lem:block-by-block-limit-exists}, 
\ref{lem:block-by-block-upper-bound}, and \ref{lem:block-by-block-lower-bound}.

\section{Proof of Lemma \ref{lem:limiting-distribution}}
\label{app:limiting-distribution-proof}
\begin{proof}
We begin by defining a function $n_M(M^\infty)$ that specifies
the truncated source $M^{n_M(M^\infty)}$.
\begin{enumerate}
\item If $M^\infty \notin \mathcal{M}_k(z^\infty)$ then
select only the first bit $n_M(M^\infty) = 1$.

\item Suppose $M^\infty \in \mathcal{M}_k(z^\infty)$.
Then let $\{L_i\}$, $\{n_i\}$, and $\widehat{p}(M^L,z^k)$ be specified so that
\eqref{eq:m-definition} is satisfied.  

Let $C_{M^\infty}$ be the optimal conditional Huffman code for $\widehat{p}(M^L | z^k)$,
and let $B_{M^\infty}$ be the number of bits required to describe
$C_{M^\infty}$ to a decoder.

We then choose $i(M^\infty) \in \Z^+$ sufficiently large so that the following
three conditions are satisfied:
\begin{description}
\item[C1] $ \left| \E{L}_{\widehat{p}} - \E{L}_{\widehat{p}_{n_{i(M^\infty)}}} \right| < \delta$,
where $\widehat{p}_{n_i}$ is the empirical distribution of $(M^L,z^k)_1^{n_i}$.  This can be satisfied
because $\widehat{p}$ is a limiting empirical distribution on the points $\{n_i\}$.
\item[C2] $ \E{\ell(C(M^L | z^k))}_{\widehat{p}_{n_{i(M^\infty)}}} \leq H_{\widehat{p}}(M^L | z^k) + 1$.
This is possible since we know that $\E{\ell(C(M^L | z^k))}_{\widehat{p}} \leq H_{\widehat{p}}(M^L | z^k) + 1$,
and that $\widehat{p}_{n_i} \rightarrow \widehat{p}$.
\item[C3] $ \frac{B_{M^\infty}}{n_{i(M^\infty)}} < \delta$.  This is possible simply
because $B_{M^\infty}$ is finite.
\end{description}

Armed with $i(M^\infty)$ we may define the following:
\begin{itemize}
\item $n(M^\infty) = n_{i(M^\infty)}$ is the relevant index in the partition sequence $\{(M^L,z^k)_i\}$.
\item $n_M(M^\infty) = \sum_{j=1}^{n_{i(M^\infty)}} L_j$ is the relevant index in the source sequence $M^\infty$.
\item $n_z(M^\infty) = k n_{i(M^\infty)}$ is the relevant index in the noise sequence $z^\infty$.
\end{itemize}
\end{enumerate}

We now describe a variable-length source coding scheme
for this constructed source $M_1^{n_M(M^\infty)}$ with
encoder-side-information $M_{n_M(M^\infty)+1}^\infty$:

\begin{enumerate}

\item If $M^\infty \notin \mathcal{M}_k$, directly encode
$M_1^{n_M(M^\infty)}$.  Recall that since
$n_M((\mathcal{M}_k)^C) = 1$, this is just the first bit $M_1$.

\item If $M^\infty \in \mathcal{M}_k$: First, specify $C_{M^\infty}$ with
the first $B_{M^\infty}$ bits in the encoding.  Then,
apply $C_{M^\infty}$ to each of the $n(M^\infty)$ blocks
in the $k$-partition $(M^L,z^k)_1^{n(M^\infty)}$.
\end{enumerate}
Call this encoding function $F(M^{n_M(M^\infty)})$.  By
Shannon's source coding converse, the expected length of this encoding must
exceed the entropy of the source $M^{n_M(M^\infty)}$.
We demonstrate that for this to hold $\mathcal{M}_k$
must be of measure zero.

\begin{eqnarray*}
0 & \leq & \E{\ell(F(M_1^{n_M(M^\infty)}))} - H(M_1^{n_M(M^\infty)}) 
\\
& \leqlabel{a} & \E{\ell(F(M_1^{n_M(M^\infty)}))} - \E{n_M(M^\infty)} 
\\
& \eqlabel{b} & \P{M^\infty \notin \mathcal{M}_k} + 
\P{M^\infty \in \mathcal{M}_k} \E{B_{M^\infty} + 
n(M^\infty) \E{\ell(C_{M^\infty}(M^L|z^k))}_{\widehat{p}_{n(M^\infty)}} 
\mid M^\infty \in \mathcal{M}_k} 
\\
& & - \P{M^\infty \notin \mathcal{M}_k} - \P{M^\infty \in \mathcal{M}_k}\E{n_M(M^\infty)
\mid M^\infty \in \mathcal{M}_k }\\
& = & \P{M^\infty \in \mathcal{M}_k} \E{B_{M^\infty} + 
																n(M^\infty) \E{\ell(C_{M^\infty}(M^L|z^k))}_{\widehat{p}_{n(M^\infty)}}
																-n_M(M^\infty) \mid M^\infty \in \mathcal{M}_k} \\
& \leqlabel{c} & \P{M^\infty \in \mathcal{M}_k} \E{ 
																n(M^\infty) \left(\delta + \E{\ell(C_{M^\infty}(M^L|z^k))}_{\widehat{p}_{n(M^\infty)}}
																            \right)
																-n_M(M^\infty) \mid M^\infty \in \mathcal{M}_k} \\
& \leqlabel{d} & \P{M^\infty \in \mathcal{M}_k} \E{ 
																n(M^\infty) \left(\delta + \E{\ell(C_{M^\infty}(M^L|z^k))}_{\widehat{p}_{n(M^\infty)}}
																           	\right)
																-n(M^\infty)\E{L}_{\widehat{p}_{n(M^\infty)}}  \mid M^\infty \in \mathcal{M}_k} 
																\\
& \leqlabel{e} & \P{M^\infty \in \mathcal{M}_k} \E{ 
																n(M^\infty) \left(\delta + H_{\widehat{p}}(M^L|z^k) + 1	\right)
																-n(M^\infty)\E{L}_{\widehat{p}_{n(M^\infty)}} \mid M^\infty \in \mathcal{M}_k} 
																\\
& \leqlabel{f} & \P{M^\infty \in \mathcal{M}_k} \E{ 
																n(M^\infty) \left(2\delta + H_{\widehat{p}}(M^L|z^k) + 1	 \right)	 											
																-n(M^\infty) \E{L}_{\widehat{p}} 
																                \mid M^\infty \in \mathcal{M}_k} 
																                \\
0 & \leqlabel{g} & \P{M^\infty \in \mathcal{M}_k} \E{ 
																n(M^\infty) \left(H_{\widehat{p}}(M^L|z^k) + 1	 	 											
																-\E{L}_{\widehat{p}} \right)
																                 \mid M^\infty \in \mathcal{M}_k}
\end{eqnarray*}
Step (a) holds by Lemma \ref{lem:selection}.  Step (b) follows from an expansion
 of both expectations.  (c) is due to condition {\bf C3}.  
 The definition of $n_M(M^\infty)$ yields (d).  Condition {\bf C2} implies (e).
 (f) follows from condition {\bf C1}. 
 Because $\delta$ can be arbitrarily small, (g) holds.
 Finally, by the definition of $\mathcal{M}_k$ in \eqref{eq:m-definition}, this last inequality
 can only be satisfied if $\P{M^\infty \in \mathcal{M}}$ is zero.
\end{proof}

\section{Proof of Lemma \ref{lem:markov-relation}}
\label{app:markov-relation-proof}
First, some notational conveniences are defined.
Suppose the given repetition scheme $\overline{\mathcal{F}}$ is applied to an
iid Bernoulli$(1/2)$ source $M^\infty$
and fixed noise sequence $z^\infty$.  Let $\{L_i\}$ be the number of source bits
decoded in each block.  Consider the $i$th block --- i.e.
$\mathcal{F}\left( (M_{\sum_{j=1}^{i-1}L_i + 1}^{\sum_{j=1}^{i-1}L_i+N},0^\infty),
z_{(i-1)N+1}^{iN}\right)$.
Let $M_{(i)}^{N} = M_{\sum_{j=1}^{i-1}L_i + 1}^{\sum_{j=1}^{i-1}L_i+N}$ indicate the source
 bits used by the encoder, 
 let $\widehat{M}_{(i)}^{L_i} = \widehat{M}_{\sum_{j=1}^{i-1}L_i + 1}^{\sum_{j=1}^{i}L_i}$
indicate the estimate produced, 
let $R_i = \frac{L_i}{N}$ indicate the rate of the block,
and let $z_{(i)}^{N} = z_{(i-1)N+1}^{iN}$ denote the noise
of the block. Finally, let $u_{(i)}^{N}=u_{(i-1)N+1}^{iN}$ denote the feedback
during the block.

Observe that the $j$th block in a repetition scheme can only affect a future block
in one way: by adjusting $L_j$ and thereby changing $M_{(i)}^{N}$ for $i > j$. 
As such, for all $i > j$, the Markov chain 
 $u_{(i)}^{N} - M_{(i)}^{N} - (u_{(j)}^{N},M_{(j)}^{N})$
holds.  The joint distribution of $(u_{(j)}^{N},u_{(i)}^{N},M_{(j)}^{N},M_{(i)}^{N})$
may then be written as
\[
p\left( u_{(j)}^{N},u_{(i)}^{N},M_{(j)}^{N},M_{(i)}^{N} \right) = 
p(u^{N}_{(j)},M^{N}_{(j)})
p(M^{N}_{(i)} | u^{N}_{(j)},M^{N}_{(j)})
p(u^{N}_{(i)} | M^{N}_{(i)})
\mbox{.}
\]
Since $E_i$ and $E_j$ are deterministic functions
of $(u^{N}_{(i)},M^{N}_{(i)})$ and $(u^{N}_{(j)},M^{}_{(j)})$ respectively,
we may easily introduce them into the joint distribution:
\beqan
p\left( u_{(j)}^{N},u_{(i)}^{N},M_{(j)}^{N},M_{(i)}^{N}, E_i,E_j \right) &= &
p\left(u^{N}_{(j)},M^{N}_{(j)}\right) p\left(E_j | u^{N}_{(j)},M^{N}_{(j)}\right) \\
&&
\cdot p\left(M^{N}_{(i)} | u^{N}_{(j)},M^{N}_{(j)}\right)
p\left(u^{N}_{(i)} | M^{N}_{(i)}\right) \\
& & \cdot
p\left(E_i | u^{N}_{(i)},M^{N}_{(i)}\right)
\mbox{.}
\eeqan
This may be rephrased as
\[
p\left( u_{(j)}^{N},u_{(i)}^{N},M_{(j)}^{N},M_{(i)}^{N}, E_i,E_j \right) = 
p\left(u_{(j)}^{N}, M_{(j)}^{N}, E_j, M^{N}_{(i)} \right)
p\left(u^{N}_{(i)}, E_i | M^{N}_{(i)} \right)
\mbox{.}
\]
Summing over $(u^{N}_{(j)},u^{N}_{(i)},M^{N}_{(j)})$ this yields
\[
p\left( E_i,E_j,M^{N}_{(i)} \right) = 
p\left( E_j,M^{N}_{(i)} \right)p\left( E_i | M^{N}_{(i)} \right) =
p\left( E_j | M^{N}_{(i)} \right)
p\left( E_i | M^{N}_{(i)} \right)
p\left( M^{N}_{(i)} \right)
\mbox{,}
\]
which proves the Markov relation $E_i - M^{N}_{(i)} - E^{i-1}$.

The above arguments may be repeated verbatim to show that $T_i- M^{N}_{(i)}-T^{i-1}$.

\section*{Acknowledgments}
The authors would like to thank Yuval Lomnitz both for stimulating discussion
and for his helpful suggestions for this manuscript.
\bibliographystyle{IEEEtran} 
\bibliography{porosity}

\end{document}